\documentclass[a4paper,11pt]{article}
\usepackage{fullpage}
\usepackage{microtype} 
\usepackage{pdfpages}
\usepackage[T1]{fontenc}
\usepackage{mathtools,amsmath,xspace}
\usepackage{amsthm}
\usepackage{amssymb}
\usepackage{tabularx}
\usepackage{multirow}
\usepackage{hyperref}
\newtheorem{theorem}{Theorem} 
\newtheorem{lemma}[theorem]{Lemma}

\newcommand{\be}{\begin{eqnarray*}}
\newcommand{\ee}{\end{eqnarray*}}
\newcommand{\tri}[1]{\triangle{#1}}
\newcommand{\triopt}{T_\mathrm{opt}}
\newcommand{\area}[1]{\textsf{area}(#1)}

\newcommand{\sd}{\ensuremath{\textsf{S}}}
\newcommand{\sdt}[1]{\ensuremath{\textsf{S}(#1)\xspace}}
\newcommand{\seg}[1]{C(#1)\xspace}

\newcommand{\ct}[1][\theta]{\textsf{C}_{#1}\xspace}
\newcommand{\ray}[1][\pi]{\eta_{#1}}
\newcommand{\foot}[1][\pi]{\delta_{#1}}

\newcommand{\dirw}{\overrightarrow{u}}

\newcommand{\vis}{\textsf{Vis}}
\newcommand{\visl}[1][\alpha]{\textsf{Vis}_{#1}}

\newcommand{\que}{\ensuremath{\mathcal{Q}\xspace}}
\newcommand{\tree}{\ensuremath{\mathcal{T}\xspace}}
\newcommand{\dw}[2]{\ensuremath{w}(#1,#2)}
\newcommand{\myeps}{\ensuremath{\varepsilon}}

\newcommand{\aff}{\Phi_v}
\newcommand{\affi}[1][i]{\varphi_{#1,v}}
\newcommand{\tc}{T(v)}
\newcommand{\dc}{D(v)}
\newcommand{\rc}{R(v)}

\newif\iffull
\fulltrue

\title{Largest triangles in a polygon\thanks{This research was supported by the Institute of Information \& communications 
Technology Planning \& Evaluation(IITP) grant funded by the Korea government(MSIT)
(No. 2017-0-00905, Software Star Lab (Optimal Data Structure and Algorithmic Applications in Dynamic Geometric Environment))}}
\author{Seungjun Lee\thanks{Department of Computer Science and Engineering, Pohang University of Science and Technology, Pohang, Korea. \texttt{juny2400@postech.ac.kr}}
\and Taekang Eom\thanks{Department of Mathematics, Pohang University of Science and Technology, Pohang, Korea. \texttt{tkeom0114@postech.ac.kr}}
\and Hee-Kap Ahn\thanks{Department of Computer Science and Engineering, Graduate School of Artificial Intelligence, Pohang University of Science and Technology, Pohang, Korea. \texttt{heekap@postech.ac.kr}}
}
\begin{document}
\date{}
\maketitle

\begin{abstract}
  We study the problem of finding maximum-area triangles that can be
  inscribed in a polygon in the plane.  We consider eight versions of
  the problem: we use either convex polygons or simple polygons as the
  container; we require the triangles to have either one corner with a
  fixed angle or all three corners with fixed angles; we either allow
  reorienting the triangle or require its orientation to be fixed.  We
  present exact algorithms for all versions of the problem.
  In the case with reorientations for convex polygons with $n$
  vertices, we also present $(1-\myeps)$-approximation algorithms.
\end{abstract}

\section{Introduction}
We study the problem of finding maximum-area triangles that are
inscribed in a polygon in the plane.  When the shape of the triangle
is fully prescribed, this problem is related to \textit{the polygon
  containment problem}, which for given two polygons $P$ and $Q$, asks
for the largest copy of $Q$ that can be contained in $P$ using
rotations, translations, and scaling.
The problem is related to the problem of \textit{inscribing polygons}
if the shape is partially prescribed. In inscribing polygons, we are
given a polygon $P$ and seek to find a best polygon with some
specified number of vertices that can be inscribed in $P$ with respect
to some measures.

Problems of this flavor have a rich history and are partly motivated
by the attempt to reduce the complexity of various geometric problems,
including the shape recognition and matching problems, arising in
various applications in pattern recognition, computer vision and
computational geometry~\cite{Blaschke1917,Fleischer1992,WU1993471}.
Chapter 30.5 in the Handbook of Discrete and Computational
Geometry~\cite{DCGhandbook} provides a survey on the related works.

There has been a fair amount of work on inscribing a maximum-area
convex $k$-gon in a polygon.  A maximum-area convex $k$-gon inscribed
in a convex $n$-gon can be computed in $O(kn+n\log n)$
time~\cite{AKM87,KLU17}.  The best algorithm for computing a
maximum-area convex polygon inside a simple $n$-gon takes $O(n^7)$
time and $O(n^5)$ space~\cite{ChangYap86}.  Hall-Holt et
al.~\cite{Hall-Holt2006} gave an $O(n\log n)$-time
$O(1)$-approximation algorithm for finding a maximum-area convex
polygon inscribed in a simple $n$-gon.  Melissaratos et
al.~\cite{melissaratos1992shortest} gave an algorithm for finding a
maximum-area triangle inscribed in a simple $n$-gon in $O(n^4)$ time.
When the maximum-area triangle is restricted to have all its corners
on the polygon boundary, and it takes $O(n^3)$ time.

For finding a maximum-area copy of a given polygon $Q$ that can be
inscribed in a polygon $P$, there are results known for cases of
convex, orthogonal, and simple polygons, possibly with holes.  A
maximum-area copy of a convex $k$-gon that can be inscribed in a
convex $n$-gon can be computed in $O(n+k\log k)$ time~\cite{ST94,GP13}
under translation and scaling, and in $O(nk^2\log k)$
time~\cite{AAS98} under translation, scaling, and rotation.  For a
maximum-area homothet\footnote{Two shape are homothetic if one can be
  obtained from the other by scaling and translation.} of a given
triangle inscribed in a convex polygon $P$ with $n$ vertices,
Kirkpatrick and Snoeyink gave an $O(\log n)$-time algorithm to find
one~\cite{kirkpatrick1995tentative}, given the vertices are stored in
order along the boundary in an array or balanced binary search tree.
The maximum-area equilateral triangles of arbitrary orientation
inscribed in a simple $n$-gon can be computed in $O(n^3)$
time~\cite{depano1987finding}.

There also have been works on finding a maximum-area partially
prescribed shape that can be inscribed in a
polygon. 
Amenta showed that a maximum-area axis-aligned rectangle inscribed in
a convex $n$-gon can be found in linear time by phrasing it as a
convex programming problem~\cite{Amenta1994}.  When the vertices are
already stored in order along the boundary in an array or balanced
binary search tree, the running time was improved to
$O(\log^2n)$~\cite{Fischer1994}, and then to
$O(\log n)$~\cite{Alt1995}.  Cabello et al.~\cite{cabello2016finding}
considered the maximum-area and maximum-perimeter rectangle of
arbitrary orientation inscribed in a convex $n$-gon, and presented an
$O(n^3)$-time algorithm.  Very recently, Choi et al.~\cite{Choi2019}
gave $O(n^3\log n)$-time algorithm for finding maximum-area rectangles
of arbitrary orientation inscribed in a simple $n$-gon, possibly with
holes. However, little is known for the case of partially prescribed
triangles inscribed in convex and simple polygons, except a PTAS
result by Hall-Holt et al.~\cite{Hall-Holt2006} for finding a
maximum-area \emph{fat}\footnote{A triangle is $\delta$-fat if all
  three of its angles are at least some specific constant $\delta$.}
triangle that can be
inscribed in a simple $n$-gon.

\newcolumntype{f}{>{\hsize=.5\hsize}X}
\newcolumntype{j}{>{\centering\arraybackslash\hsize=.65\hsize}X}
\newcolumntype{g}{>{\centering\arraybackslash\hsize=\hsize}X}
\newcolumntype{s}{>{\centering\arraybackslash\hsize=.36\hsize}X}
\subsection{Our results}
\renewcommand{\arraystretch}{1.1}
\begin{table}
  \begin{tabularx}{\textwidth}{f||j|g||s|s}
    \hline & \multicolumn{2}{c||}{Convex polygons} &
    \multicolumn{2}{c}{Simple polygons} \\ \hline
    Fixed angles & All &
    One & All & One \\ \hline \hline
    \multirow{2}{*}{Axis-aligned} & {$O(\log n)$~\cite{kirkpatrick1995tentative}}
    & \multirow{2}{*}{$O(\log n)$} & {$O(n \log n)$} &
    \multirow{2}{*}{$O(n^2 \log n)$} \\ 
    & (homothet) &  & (homothet) & \\ \hline
    \multirow{2}{*}{Reorientations} & {$O(n^2)$}
     & $O(n^3)$ & \multirow{2}{*}{$O(n^2 \log n)$} & \multirow{2}{*}{$O(n^4)$} \\
 & 
    {$O(\myeps^{-\frac{1}{2}}\log n + \myeps^{-1})$} & 
    {$O(\myeps^{-\frac{1}{2}}\log{n}+\myeps^{-1} \log{\myeps^{-\frac{1}{2}}})$} &
     &  \\
    \hline
  \end{tabularx}
  \caption{Time complexities of the algorithms. All algorithms use $O(n)$ space.}
  \label{tab:results}
\end{table}
We study the problem of finding maximum-area triangles that can be
inscribed in a polygon in the plane.  We consider eight versions of
the problem: we use either convex polygons or simple polygons as the
container; we require the triangles to have either one corner with a
fixed interior angle or all three corners with fixed interior angles;
we either allow reorienting the triangle or require its orientation to
be fixed.  We study all versions of the problem in this paper
and present efficient algorithms for them.
Table~\ref{tab:results} summarizes our results.

We assume that the vertices of the input polygon are stored in order
along its boundary in an array or a balanced binary search tree.  We
say a triangle is \emph{axis-aligned} if one of its sides is parallel
to the $x$-axis, and we call the side the \emph{base} of the triangle.
We say a triangle has one fixed angle if one of the two interior
angles at corners incident to the base of the triangle is fixed.

For a convex polygon $P$ with $n$ vertices, a maximum-area homothet of
a given triangle that can be inscribed in $P$ can be computed in
$O(\log n)$ time~\cite{kirkpatrick1995tentative}. For axis-aligned
triangles with one fixed angle, we present an algorithm that computes
a maximum-area such triangle that can be inscribed in $P$ in
$O(\log n)$ time using $O(n)$ space.

When reorientations are allowed, we present an algorithm that computes
a maximum-area triangle with fixed interior angles that can be
inscribed in $P$ in $O(n^2)$ time using $O(n)$ space.  We also present
an $(1-\myeps)$-approximation algorithm that takes
$O(\myeps^{-\frac{1}{2}}\log n + \myeps^{-1})$ time.  For triangles
with one fixed angle, we present an algorithm to compute a
maximum-area triangle that can be inscribed in $P$ in $O(n^3)$ time
using $O(n)$ space. We also present an $(1-\myeps)$-approximation
algorithm that takes
$O(\myeps^{-\frac{1}{2}}\log{n}+\myeps^{-1} \log{\myeps^{-\frac{1}{2}}})$ time.

For a simple polygon $P$ with $n$ vertices, we present an algorithm
that computes a maximum-area homothet of a given triangle that can be
inscribed in $P$ in $O(n \log n)$ time using $O(n)$ space.  We also present an algorithm
to compute a maximum-area axis-aligned triangle with one fixed angle
that can be inscribed in $P$ in $O(n^2 \log n)$ time using $O(n)$ space.

When reorientations are allowed, we present an algorithm to compute a
maximum-area triangle with fixed interior angles that can be inscribed
in $P$ in $O(n^2 \log n)$ time using $O(n)$ space.  For triangles with one fixed angle,
we present an algorithm to compute a maximum-area triangle that can be
inscribed in $P$ in $O(n^4)$ time using $O(n)$ space.

\bigskip
Whenever we say a \emph{largest triangle}, it refers to a maximum-area
triangle inscribed in $P$.  We denote the triangle with three corners
$p, q, r$ (counterclockwise order) by $\tri{pqr}$, where $pq$ is base.
For two fixed angles $\alpha,\beta >0$, we call the triangle with
$\angle{rpq}=\alpha$ \emph{an $\alpha$-triangle} and the triangle with
$\angle{rpq}=\alpha$ and $\angle{pqr}=\beta$ \emph{an
  $(\alpha,\beta)$-triangle}.  Let $\area{T}$ denote the area of a
triangle $T$.

\section {Largest Triangles in a Convex Polygon}
Consider a convex polygon $P$ with $n$ vertices in the plane. We show
for fixed angles $\alpha$ and $\beta$ 
how to find largest $(\alpha,\beta)$-triangles and largest
$\alpha$-triangles, aligned to the $x$-axis or of arbitrary
orientation, that can be inscribed in $P$.

\subsection{Largest
  \texorpdfstring{$(\alpha,\beta)$}{(a,b)}-triangles}
Since all interior angles of an $(\alpha,\beta)$-triangle are fixed,
this problem is to find a largest copy of a given triangle that can be
inscribed in $P$ using rotations, translations, and scaling.  When the
orientation of triangles is fixed, the problem reduces to finding a
largest \emph{homothet} of a given triangle that can be inscribed in
$P$.  A \emph{homothet} of a figure is a scaled and translated copy of
the figure.

A largest homothet of a given triangle that can be inscribed in $P$ 
can be computed in $O(\log n)$ time~\cite{kirkpatrick1995tentative}.  Thus, we focus on
the case in which arbitrary orientations are allowed.
This problem is similar to finding a largest equilateral triangle in a
convex polygon~\cite{depano1987finding}.  A largest
$(\alpha,\beta)$-triangle in a convex polygon $P$ must have at least
one corner lying on a vertex of $P$ by the same argument for largest
equilateral triangles in Theorem~1 of~\cite{depano1987finding}.

Consider an $(\alpha,\beta)$-triangle $\tri{t_0t_1t_2}$.
Let $\aff(s,\vartheta)$ denote the affine transformation that scales
$s$ and rotates $\vartheta$ in counterclockwise direction around a
vertex $v$ of $P$.
Let\[
  \affi[0]=\aff(\frac{\sin\beta}{\sin(\alpha+\beta)},\alpha),\,\affi[1]=\aff(\frac{\sin(\alpha+\beta)}{\sin\alpha},\beta),
  \,\affi[2]=\aff(\frac{\sin\alpha}{\sin\beta},\pi-\alpha-\beta).\]
\noindent
For $t_i$ lying on a vertex $v$ of $P$, we observe that
$t_{i+1}, t_{i+2}\in P$ if and only if $t_{i+2}\in P\cap\affi(P)$ with
indices under modulo 3. See Figure~\ref{fig:fixed_angle}(a) for an illustration.
Thus, for a fixed vertex $v$ of $P$, we can compute a largest triangle
$\tri{t_0t_1t_2}$ with $t_0$ at $v$ in $O(n)$ time by finding the
longest segment $vt_{i+2}$ contained in $P\cap\affi(P)$.
By repeating this for every vertex $v$ of $P$ such that corner $t_i$
lies on $v$ for $i=0,1,2$, a largest $(\alpha,\beta)$-triangle can be
computed in $O(n^2)$ time.

\begin{theorem}
  Given a convex polygon $P$ with $n$ vertices in the plane and two
  angles $\alpha, \beta$, we can compute a maximum-area
  $(\alpha,\beta)$-triangle of arbitrary orientations that can be
  inscribed in $P$ in $O(n^2)$ time.
\end{theorem}

We give an example of a convex polygon with $\Omega(n^2)$
combinatorially distinct candidates of an optimal triangle
with $\alpha=\beta=60^\circ$.
Let $P$ be a convex polygon with $3n$ vertices such that
$2n$ vertices of $P$ are placed uniformly on a circular arc
of interior angle $120^{\circ}$, and the remaining $n$ vertices are
placed densely in the neighborhood of the center of the arc as shown
Figure~\ref{fig:fixed_angle}(b).  If $v$ is one of the $n$ vertices
near the center, $P\cap\affi[0](P)$ has $2n-1$ vertices
along by the arc,
and thus there are $\Theta(n)$ candidates of
$t_2$ for each such vertex $v$ to consider for the longest $vt_{2}$ in
$P\cap\affi[0](P)$. This gives $\Theta(n^2)$ candidates for the
longest $vt_{2}$ of similar lengths, and thus $\Theta(n^2)$
combinatorially distinct $(\alpha,\beta)$-triangles with side $vt_2$.
Any algorithm iterating over all such triangles takes $\Omega(n^2)$
time.

\begin{figure}[ht]
  \begin{center}
    \includegraphics[width=.9\textwidth]{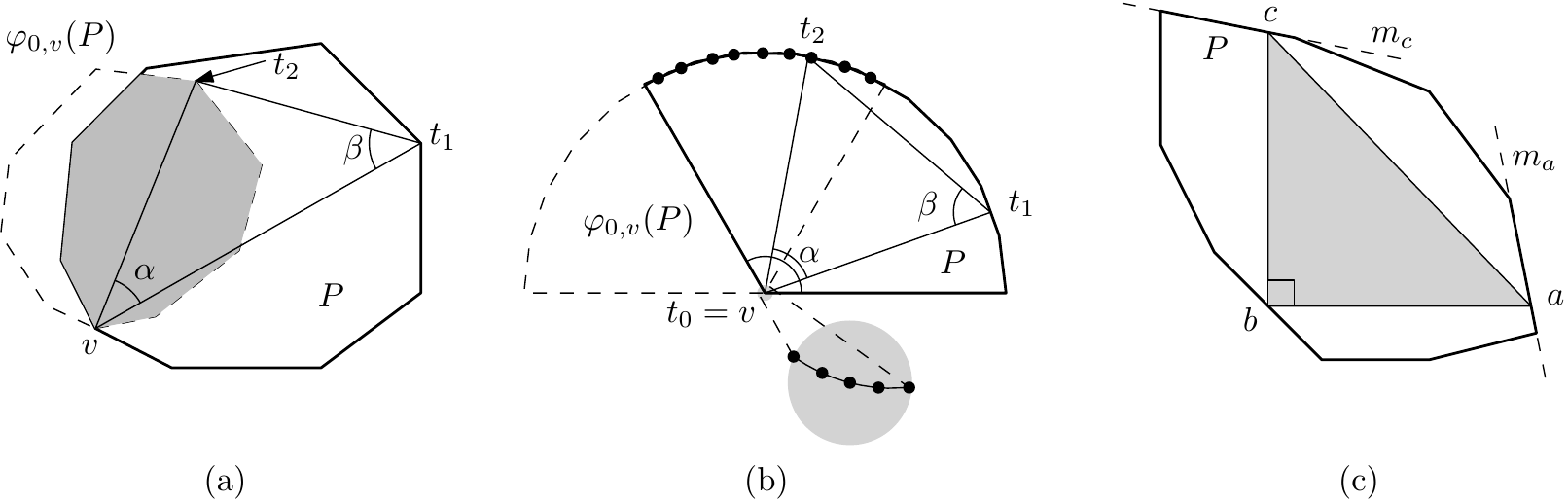}
  \end{center}
  \caption{(a) A largest $(\alpha, \beta)$-triangle $\tri{t_0t_1t_2}$
    with $t_0$ lying on a vertex $v$ of $P$
    is determined by the longest segment $vt_{2}$ contained
    in $P\cap\affi[0](P)$.
    (b) An example of $\Theta(n^2)$ combinatorially distinct
    $(\alpha,\beta)$-triangles with side $vt_2$
    to consider for an optimal triangle.  
  (c) The largest axis-aligned right triangle with $m_a<0$ and $m_c<0$.}
  \label{fig:fixed_angle}
\end{figure}

\subsection{Largest \texorpdfstring{$\alpha$}{a}-triangles}
This problem is to find a largest triangle with one corner angle fixed
to a constant $\alpha$ that can be inscribed in a convex polygon
$P$. We consider $\alpha$-triangles that are either axis-aligned or of
arbitrary orientations.

\subsubsection{Largest axis-aligned \texorpdfstring{$\alpha$}{a}-triangles}
We start with an algorithm to compute a largest axis-aligned
$\alpha$-triangle for $\alpha=90^\circ$.  Alt et
al.~\cite{alt1995computing} presented an algorithm of computing a
largest axis-aligned rectangle that can be inscribed in a convex
polygon. We follow their approach with some modification.  If two
corners of the triangle are on the polygon boundary, the algorithm by
Alt et al., works to compute a largest axis-aligned right triangle
that can be inscribed in $P$.

So in the following, we focus on the case that a largest axis-aligned
right triangle has all its corners on the boundary of $P$.  Consider a
largest axis-aligned right triangle $\tri{bac}$ with all three
  corners on the boundary of $P$.  See
Figure~\ref{fig:fixed_angle}(c) for an illustration.  Let $m_a$ and
$m_c$ denote the slopes of the polygon edges where $a$ and $c$ lie,
respectively, and let $m_{ac}$ denote the slope of $ac$.
Then, either (1) $m_a<0$ or $m_c<0$, or (2) $m_a>0$ and $m_c>0$.
Observe that no rectangle containing a largest
axis-aligned right triangle belonging to case (1) is contained in
$P$, and thus the algorithm by Alt et al. fails to find a
  largest axis-aligned right triangle belonging to the case.

\begin{figure}[ht]
  \begin{center}
    \includegraphics[width=.9\textwidth]{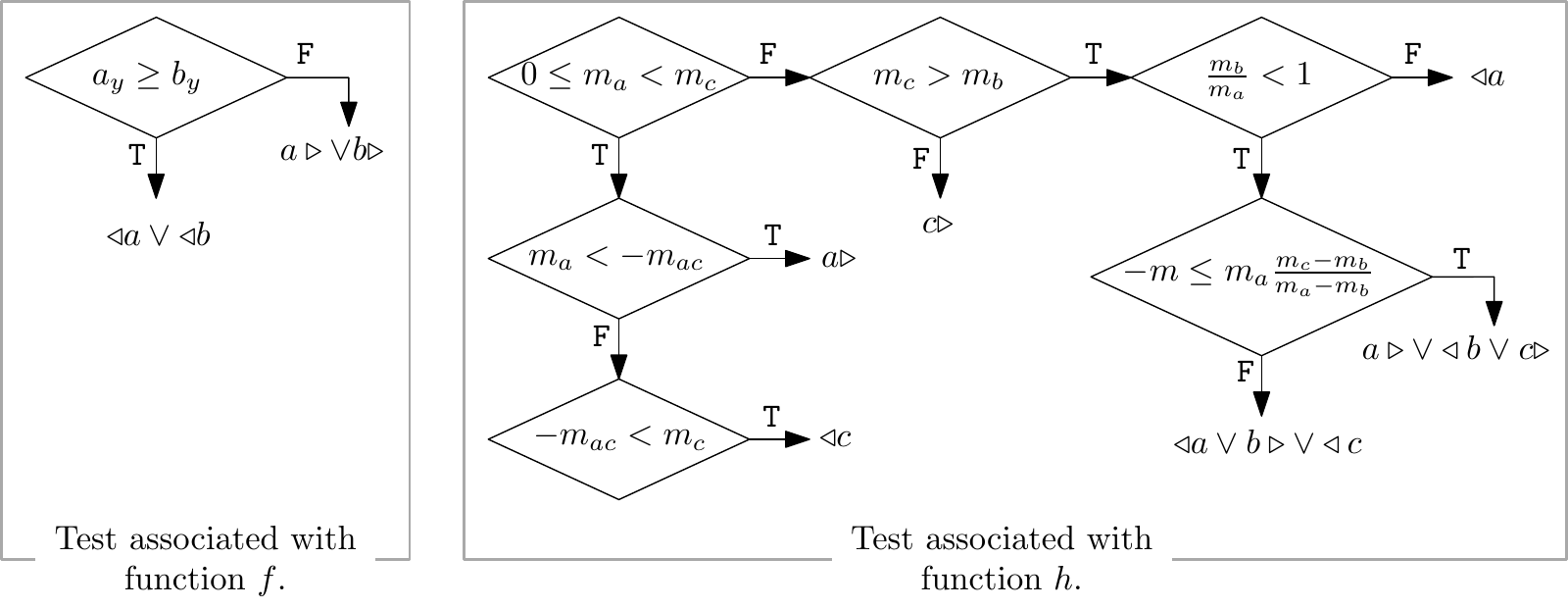}
  \end{center}
  \caption{Modified tests associated with functions $f$ and $h$.}
  \label{fig:test_h}
\end{figure}

We compute a largest axis-aligned right triangle in $O(\log n)$ time
using the tentative prune-search algorithm~\cite{alt1995computing} by
replacing the tests associated with functions $f$ and $h$ by the ones
in Figure~\ref{fig:test_h}.  Using the tests, we can determine a half
of the candidate triples of polygon edges in which a largest
axis-aligned right triangle cannot have their corners, and continue to
find a largest axis-aligned right triangle on the remaining half of
the candidate triples of polygon edges.

The algorithm to compute a largest axis-aligned right triangle can
compute the axis-aligned $\alpha$-triangle using linear
transformation $L_\alpha=\big(\begin{smallmatrix}
  1 & \cot\alpha\\
  0 & 1
\end{smallmatrix}\big)^{-1}$ for $\alpha$. Then, $L_\alpha(T)$ is an
axis-aligned right triangle for an $\alpha$-triangle $T$. Observe that
$T$ is a largest axis-aligned $\alpha$-triangle inscribed in $P$ if and only
if $L_\alpha(T)$ is a largest axis-aligned right triangle inscribed in
$L_\alpha(P)$, a convex polygon.
However, it takes $O(n)$ time for
computing entire description of $L_\alpha(P)$.
To reduce the time complexity, we
compute $L_\alpha(p)$ only when we need the slope of
the polygon edge containing $p$ or
the right side of an $\alpha$-triangle
with corner at $p$ on the polygon boundary.
Since there are $O(\log n)$ decision steps in the algorithm and
each decision step uses only a
constant number of points on the polygon boundary,
we have following theorem.

\begin{theorem}
  Given a convex polygon $P$ of $n$ vertices in the plane and
  an angle $\alpha$, we can find a maximum-area
  axis-aligned $\alpha$-triangle that can be inscribed in $P$ in
  $O(\log n)$ time.
\end{theorem}

\subsubsection{Largest \texorpdfstring{$\alpha$}{a}-triangles of arbitrary orientations}
We can compute a largest $\alpha$-triangle of arbitrary orientations
by simply iterating over all triples of edges of $P$. 
For each triple of edges of $P$, we can find a largest $\alpha$-triangle $T$
with corners on the edges of the triple in $O(1)$ time. 
\begin{theorem}
  Given a convex polygon $P$ of $n$ vertices in the plane and an angle
  $\alpha$, we can find a maximum-area $\alpha$-triangle of arbitrary orientation 
  that can be inscribed in $P$ in $O(n^3)$ time using $O(n)$ space.
\end{theorem}

One may wonder if the running time can be improved.
Cabello~et~al. showed a construction of a convex polygon with $n$ vertices
that has $\Theta(n^3)$ combinatorially distinct rectangles
  that can be inscribed in the polygon.
By using a similar construction,
we can show that there are $\Theta(n^3)$ combinatorially distinct 
$\alpha$-triangles that can be inscribed in a convex polygon
with $n$ vertices. 
Thus, any algorithm iterating
over all those combinatorially distinct triangles takes $\Omega(n^3)$ time.

\subsection{FPTAS in arbitrary orientations}
Let $\triopt$ be a largest $(\alpha,\beta)$-triangle that can be
inscribed in $P$. We can compute an $(\alpha,\beta)$-triangle
inscribed in $P$ 
whose area is at least $(1-\myeps)$ times $\area{\triopt}$
in $O(\myeps^{-\frac{1}{2}}\log n + \myeps^{-1})$ time using
$\myeps$-kernel~\cite{cabello2016finding}.  For any $\myeps\in(0,1)$, an
\emph{$\myeps$-kernel} for a convex polygon $P$ is a convex polygon
$P_\myeps$ such that for all unit vectors $u$ in the plane,
$(1-\myeps)\dw{P}{u}\leq \dw{P_\myeps}{u}$, where $\dw{P}{u}$ is the
length of the orthogonal projection of $P$ onto any line parallel to
$u$.
\begin{theorem}
  Given a convex polygon $P$ with $n$ vertices in the plane,
  two angles $\alpha, \beta$, 
  and $\myeps>0$, we can find an
  $(\alpha,\beta)$-triangle that can be inscribed in $P$
  and
  whose area is at least 
  $(1-\myeps)$ times the area of a maximum-area
  $(\alpha,\beta)$-triangle inscribed in $P$ in 
  $O(\myeps^{-\frac{1}{2}}\log n + \myeps^{-1})$ time.
\end{theorem}
\begin{proof}
  By Lemma 1 in \cite{cabello2016finding}, an $\myeps$-kernel
  $P_\myeps$ of $P$ has $O(\myeps^{-\frac{1}{2}})$ vertices and it can
  be computed in $O(\myeps^{-\frac{1}{2}}\log n)$ time.  A largest
  $(\alpha,\beta)$-triangle inscribed in $P_{\frac{\myeps}{32}}$ has
  area at least $(1-\myeps)$ times the area of the largest
  $(\alpha,\beta)$-triangle inscribed in $P$ by Lemma 8 in
  \cite{cabello2016finding}, and it can be computed in
  $O(\myeps^{-1})$ time.
\end{proof}

\subsubsection{Largest \texorpdfstring{$\alpha$}{a}-triangles of arbitrary orientations}
\label{sec:convex apx alpha arbitrary}
Let $\triopt$ denote a largest $\alpha$-triangle that can be
inscribed in $P$.
We can compute an $\alpha$-triangle inscribed in $P$ whose area
is at least $(1-\myeps)$ times $\area{\triopt}$
in $O(\myeps^{-\frac{1}{2}}\log n + \myeps^{-\frac{3}{2}})$
time using the algorithm by Cabello et al.~\cite{cabello2016finding}.
We can improve the time complexity further to
$O(\myeps^{-\frac{1}{2}}\log n + \myeps^{-1}\log\myeps^{-\frac{1}{2}})$ using
the approximation method by Ahn et al.~\cite{ahn2007maximizing}.  

\iffull
We use $d$ to denote the \emph{diameter} of a convex polygon $P$ which is
the maximum distance between any two points in $P$, and 
$w$ to denote the \emph{width} of $P$ which is the minimum distance between
two parallel lines enclosing $P$.
Let $c_1 =\min \{\frac{1}{16}, \frac{\cot{(\alpha/2)}}{4}, 
\frac{|\tan\alpha|}{4}\}$ and $c_2=2\min\{\frac{1-\cos\alpha}{\alpha},
\frac{1+\cos\alpha}{\pi-\alpha}\}$ be the constants defined by $\alpha$. 
\fi
\iffull
\begin{lemma}
  \label{lem:approx_1}
  $\area{\triopt}\geq c_1 dw$.
\end{lemma}
\begin{proof}
  Let $pq$ be a diameter of $P$, and let $R$ be a rectangle
  circumscribed to $P$ with two sides parallel to $pq$ such that $P$
  touches all four sides of $R$.  Let $w'$ be the side of $R$
  orthogonal to $pq$. See Figure \ref{fig:approx_1}(a).
  
  Consider two interior-disjoint triangles with $pq$ as the base
  and total height $w'$ that are inscribed in $P$.
  Let $\tri{pqr}$ be one of the triangles whose
  height is at least $\frac{w'}{2}$.  Without loss of generality,
  assume that the bisecting line of $pq$ intersects the boundary of
  $qr$ at $s$ and let $\angle pqs=\gamma$.
  
  If $\alpha<\gamma$ or $\alpha>\pi-2\gamma$, then either $\tri{qtp}$
  (with $\area{\tri{qtp}}=\frac{\tan\alpha}{4}d^2$) or $\tri{tpq}$
  (with $\area{\tri{tpq}}=\frac{\cot(\alpha/2)}{4}d^2$)
  is an $\alpha$-triangle, where $t$ is the point on the
  bisecting line of $pq$ achieving $\angle pqt = \alpha$ or
  $\angle qtp=\alpha$. See Figure~\ref{fig:approx_1}(b).
  If $\gamma\leq\alpha\leq\pi-2\gamma$, $\tri{tqs}$
  (with $\area{\tri{tqs}}\geq\frac{1}{16}dw$) is an
  $\alpha$-triangle inscribed in $P$, where $t$ is the point on $pq$
  achieving $\angle{stq}=\alpha$ or
  $\angle{qst}=\alpha$, while satisfying
  $\angle{stq}\le \frac{\pi-\gamma}{2}$.
  See Figure \ref{fig:approx_1}(c).  Since $w\leq d$, the lemma
  holds.
  \begin{figure}[ht]
    \begin{center}
      \includegraphics[width=.8\textwidth]{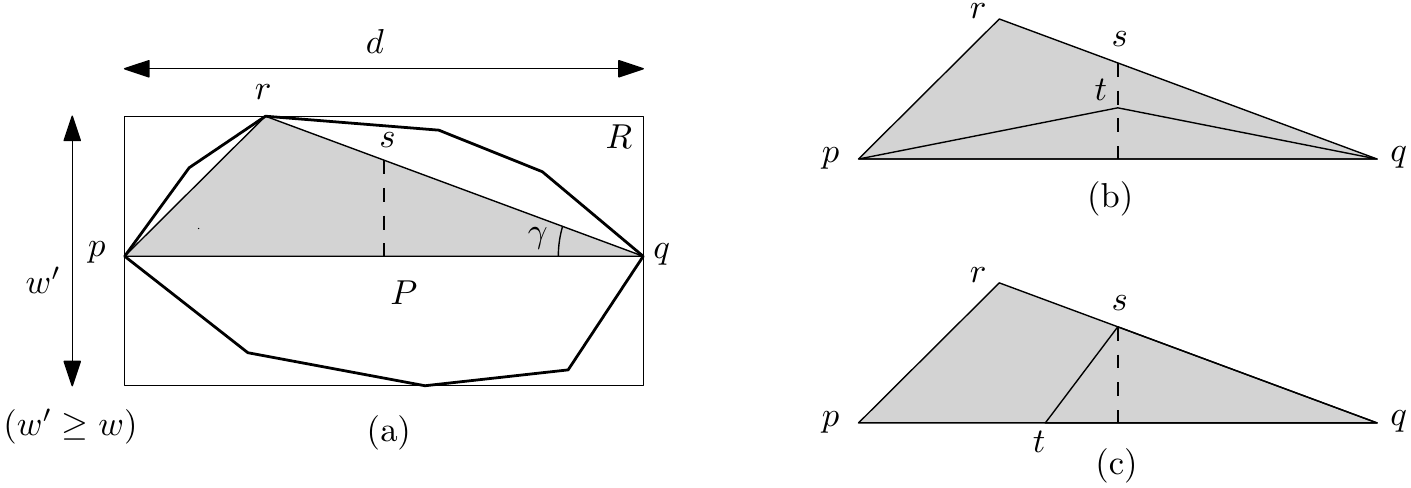}
    \end{center}
    \caption{(a) Proof of Lemma \ref{lem:approx_1}.  (b)
      $\alpha<\gamma$ or $\alpha>\pi-2\gamma$.  (c)
      $\gamma\leq\alpha\leq\pi-2\gamma$.}
    \label{fig:approx_1}
  \end{figure}
\end{proof}
\fi
\iffull
Let $\dirw$ be one of the directions of the lines defining the width
of $P$.  Let $\vartheta$ be the angle from $\dirw$ to the ray from $p$
bisecting $\angle rpq$ for a largest $\alpha$-triangle $\tri{pqr}$
inscribed in $P$ in counterclockwise direction, where
$-\pi\leq\vartheta\leq\pi$.  Then we have the following technical lemma.
\begin{lemma}
  \label{lem:approx_2}
  $\min\{\left||\vartheta|-\frac{\alpha}{2}\right|,\left|\pi-\frac{\alpha}{2}-|\vartheta|\right|\}\leq\frac{w}{c_1
    c_2 d}$.
\end{lemma}
\begin{proof}
  Clearly, $\triopt$ is contained in the strip of $w$.
  Then $\area{\tri{pqr}}\leq
  \frac{1}{2}\frac{w}{|\sin(|\vartheta|-(\alpha/2))|}\cdot
  \frac{w}{|\sin(|\vartheta|+(\alpha/2))|}=
  \frac{w^2}{|\cos\alpha-\cos{|2\vartheta|}|}$. See Figure
  \ref{fig:approx_2}.

  For $\theta\in[0,\frac{\pi}{2}]$,
  $|\cos\alpha-\cos2\theta|\geq2\min\{\frac{1-\cos\alpha}{\alpha},
  \frac{1+\cos\alpha}{\pi-\alpha}\}|\theta-\frac{\alpha}{2}|$. Observe
  also that the graph of $|\cos\alpha-\cos2\theta|$ is symmetric
  with respect to $\theta=\frac{\pi}{2}$.  Therefore,
  $c_1dw\leq\area{\triopt}\leq
  \frac{w^2}{|\cos\alpha-\cos{|2\vartheta|}|}$, and
  $d_{\alpha}\min
  \{\left||\vartheta|-\frac{\alpha}{2}\right|,\left|\pi-\frac{\alpha}{2}
    -|\vartheta|\right|\}\leq|\cos\alpha-\cos{|2\vartheta|}|\leq
  \frac{w}{c_1d}$.
  \begin{figure}[ht]
    \begin{center}
      \includegraphics[width=0.7\textwidth]{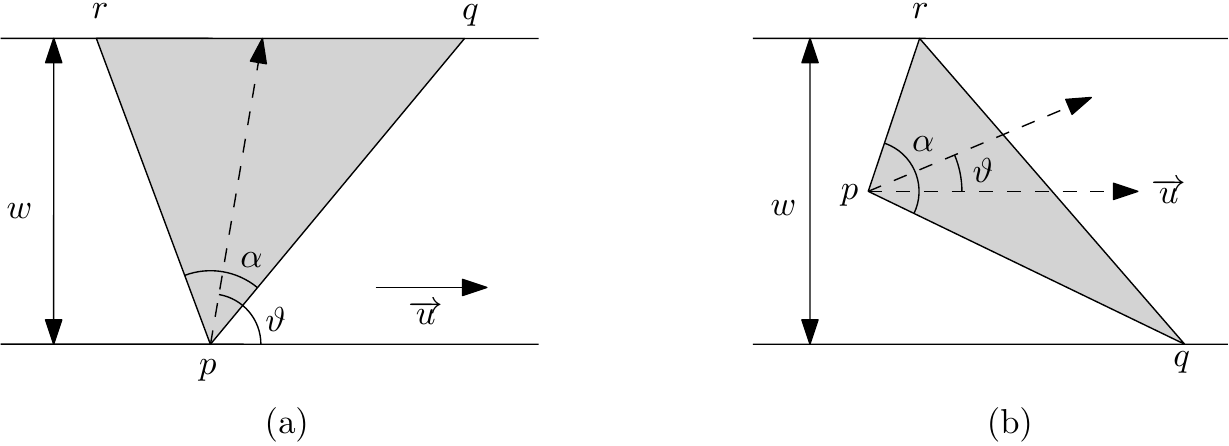}
    \end{center}
    \caption{$\area{\triopt}\leq\area{\tri{pqr}}$ if (a)
      $\frac{\alpha}{2}<|\vartheta|<\pi-\frac{\alpha}{2}$, or (b)
      $|\vartheta|<\frac{\alpha}{2}$ or
      $\pi-\frac{\alpha}{2}<|\vartheta|$.}
    \label{fig:approx_2}
  \end{figure}
\end{proof}
\fi
\iffull
Let $\triopt(\theta)(-\pi\leq\theta\leq\pi)$ denote a largest
$\alpha$-triangle $\tri{pqr}$ such that the angle from $\dirw$ to the
ray from $p$ bisecting $\angle rpq$ in counterclockwise direction is
$\theta$.
\begin{lemma}
  \label{lem:approx_3}
  Given $\myeps>0$, $\area{\triopt(\vartheta+\delta)} > (1-\myeps){\area{\triopt(\vartheta)}}$ for
  $|\delta|\leq\min\{\frac{\alpha}{2},
  \frac{\pi-\alpha}{2},\frac{c_1w}{2d}\myeps\}$.
\end{lemma}
\begin{proof}
  Let $\tri{pqr}$ be
  $\triopt(\vartheta)$. 
  Without loss of generality, assume $|pq|\leq |pr|$.  Let $p'$ and
  $q'$ be points on $pq$ and $qr$, respectively, and let
  $\angle p'rp=\angle q'pq=|\delta|$.  Also, let $t$ be the
  intersection point of $pq'$ and $rp'$. See Figure~\ref{fig:approx_3}
  for an illustration.

  If $|\delta|\leq\min\{\frac{\alpha}{2},\frac{\pi-\alpha}{2}\}$, then
  \begin{eqnarray*}
    \area{\triopt(\vartheta)}-\area{\triopt(\vartheta+\delta)}
    &\leq& \area{\tri{pqr}}-\area{\tri{tq'r}}\\ &\leq&\area{\tri{rpp'}}+
                                                       \area{\tri{pqq'}}\\
    &\leq& \frac{1}{2}\big(\frac{\sin\alpha}{\sin(\alpha+|\delta|)}|pr|^2+|pq|^2\big)
           \sin|\delta|\\
    &\leq&(|pq|^2+|pr|^2)\sin|\delta|\leq 2\max(|pq|^2,|pr|^2)\sin|\delta|.
  \end{eqnarray*}

  By Lemma \ref{lem:approx_1},
  $1-\frac{\area{\triopt(\vartheta+\delta)}}
  {\area{\triopt(\vartheta)}}\leq\frac{2\max(|pq|^2,|pr|^2)}
  {\area{\triopt(\vartheta)}}\sin|\delta|\leq
  \frac{2d^2}{c_1wd}\sin|\delta|\leq
  \frac{2d}{c_1w}|\delta|\leq\myeps$, and the lemma
  follows.
\end{proof}
\fi
\iffull
\begin{figure}[ht]
  \begin{center}
    \includegraphics[width=0.4\textwidth]{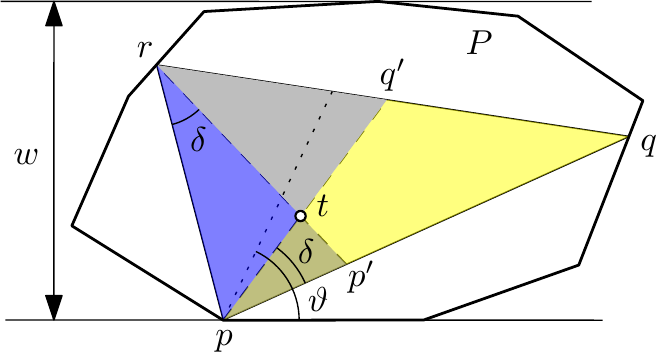}
  \end{center}
  \caption{$\area{\tri{pqr}}-\area{\tri{tq'r}}\leq\area{\tri{rpp'}}+
    \area{\tri{pqq'}}$ and
    $\area{T(\vartheta+\delta)}\geq\area{\tri{tq'r}}$ in proof of
    Lemma~\ref{lem:approx_3}.}
  \label{fig:approx_3}
\end{figure}
\fi
\begin{lemma}
  \label{lem:fptas-alpha}
  Given a convex polygon $P$ with $n$ vertices in the plane, an angle
  $\alpha$, and $\myeps>0$, we can find an $\alpha$-triangle that can
  be inscribed in $P$ and whose area is at least
  $(1-\myeps)\area{\triopt}$ in $O(\myeps^{-1}\log n)$ time.
\end{lemma}
\iffull
\begin{proof}
  We sample all orientations $\theta$ in $-\pi\leq\theta\leq\pi$ at
  interval
  $\min\{\frac{\alpha}{2},\frac{\pi-\alpha}{2}
  \frac{c_1w}{2d}\myeps\}$ that satisfy
  $\min\{\left||\theta|-
    \frac{\alpha}{2}\right|,\left|\pi-\frac{\alpha}{2}-|\theta|\right|\}\leq
  \frac{w}{c_1 c_2 d}$.

  For each sampled orientation $\theta$,
  $L_{\alpha,\theta}(\triopt(\theta))$ is a largest axis-aligned right
  triangle in $L_{\alpha,\theta}(P)$ for the linear transformation
  $L_{\alpha,\theta}= \big(\begin{smallmatrix}
    1 & \cot\alpha\\
    0 & 1
  \end{smallmatrix}\big)^{-1}R\left(\frac{\alpha}{2}-\theta\right)$, where
  $R(\phi)= \big(\begin{smallmatrix}
    \cos{\phi} & -\sin{\phi}\\
    \sin{\phi} & \cos{\phi}
  \end{smallmatrix}\big)$.
  So, we can compute $\triopt(\theta)$ in $O(\log n)$ time using the same
  technique used for the axis-aligned $\alpha$-triangles.

  We can obtain an orientation $\theta$ such that
  $\area{\triopt(\theta)}\geq(1-\myeps)\area{\triopt}$ for at least one of
  the sampled orientations by Lemmas~\ref{lem:approx_2} and
  \ref{lem:approx_3}.  The running time of the algorithm is
  $O(\myeps^{-1}\log n)$.
\end{proof}
\fi
After applying the inner approximation using an $\myeps$-$kernel$, we
can obtain the following theorem.

\begin{theorem}
  Given a convex polygon $P$ with $n$ vertices in the plane, an angle
  $\alpha$, and $\myeps>0$, we can find an $\alpha$-triangle that can be
  inscribed in $P$ and whose area is at least $(1-\myeps)$ times the area of a
  maximum-area $\alpha$-triangle inscribed in $P$ in
  $O(\myeps^{-\frac{1}{2}}\log n +
  \myeps^{-1}\log\myeps^{-\frac{1}{2}})$ time.
\end{theorem}
\begin{proof}
  By Lemma 1 in \cite{cabello2016finding}, an $\myeps$-kernel
  $P_\myeps$ of $P$ has $O(\myeps^{-\frac{1}{2}})$ vertices and it can
  be computed in $O(\myeps^{-\frac{1}{2}}\log n)$ time.  A largest
  $\alpha$-triangle in $P_{\frac{\myeps}{64}}$ has area at least
  $(1-\frac{\myeps}{2})\area{\triopt}$ by Lemma~8 in
  \cite{cabello2016finding}.  Then, an
  $(1-\frac{\myeps}{2})$-approximation to the largest
  $\alpha$-triangle in $P_{\frac{\myeps}{64}}$ is an $(1-\myeps)$
  -approximation of the largest $\alpha$-triangle in $P$. We can
  compute an $(1-\frac{\myeps}{2})$-approximation to the largest
  $\alpha$-triangle inscribed in $P_{\frac{\myeps}{64}}$ in
  $O(\myeps^{-1}\log\myeps^{-\frac{1}{2}})$ time by
  Lemma~\ref{lem:fptas-alpha}.
\end{proof}

\section{Largest \texorpdfstring{$(\alpha,\beta)$}{(a,b)}-triangles in a Simple Polygon}
In this section, we show how to find a largest $(\alpha,\beta)$-triangle that can be inscribed 
in a simple polygon $P$ with $n$ vertices in the plane. Without loss of generality, 
we assume no three vertices of $P$ are collinear.

A triangle $T$ inscribed in $P$ may touch some boundary elements (vertices and edges) of $P$.
We call an edge of $P$ that a corner of $T$ touches \emph{a corner contact} of $T$, 
and a vertex of $P$ that a side of $T$ touches in its interior \emph{a side contact} of $T$.
We call the set of all corner and side contacts of $T$ \emph{the contact set} of $T$. 
We say a triangle $T$ \emph{satisfies} a contact set $C$ 
if $C$ is the contact set of $T$.

 We use $\ray[\theta](p)$ to denote the ray emanating from $p$ 
 that makes angle $\theta$ from the positive $x$-axis in counterclockwise direction.
The inclination of line (or segment) is the angle that the line makes from the positive 
$x$-axis in counterclockwise direction.

\subsection{Largest axis-aligned \texorpdfstring{$(\alpha,\beta)$}{(a,b)}-triangles}
\label{sec:axis_aligned_fixed_angle}
Finding a largest axis-aligned $(\alpha,\beta)$-triangle is
equivalent to finding a largest homothet inscribed in $P$.
For an axis-aligned $(\alpha,\beta)$-triangle $T$ inscribed in $P$,
we use $r(T)$ to denote the left endpoint of the base of $T$, and call
it \emph{the anchor} of $T$.
For an axis-aligned $(\alpha,\beta)$-triangle $T$ inscribed in $P$ and satisfying a contact set $C$,
we say $T$ is \emph{maximal} if there is no axis-aligned $(\alpha,\beta)$-triangle of larger 
area inscribed in $P$ and satisfying a contact set $C'$ with $C\subseteq C'$.
For a point $r$ in the interior of $P$, consider the largest axis-aligned $(\alpha,\beta)$-triangle, 
denoted by $T(r)$, with $r$ at its anchor. For ease of description, 
we say $C$ is the contact set of $r$
and $r$ satisfies $C$, for the contact set $C$ of $T(r)$.
We use $\seg{r}$ to denote the contact set of $r$.
We also say $r$ is \emph{maximal} if $T(r)$ is maximal.

To compute a largest axis-aligned $(\alpha,\beta)$-triangle that can be inscribed in $P$, 
we consider all maximal axis-aligned $(\alpha,\beta)$-triangles
and find a largest triangle among them.
To find all maximal $(\alpha,\beta)$-triangles, 
we construct a subdivision of $P$ by $(\alpha,\beta)$ such that 
a maximal $(\alpha,\beta)$-triangle $T$ has $r(T)$ at a vertex of the subdivision.

\subsubsection{Subdivision of \texorpdfstring{$P$}{P} by angles \texorpdfstring{$(\alpha,\beta)$}{(a,b)}}
\begin{figure}[ht]
  \centering
  \includegraphics[width=.9\textwidth]{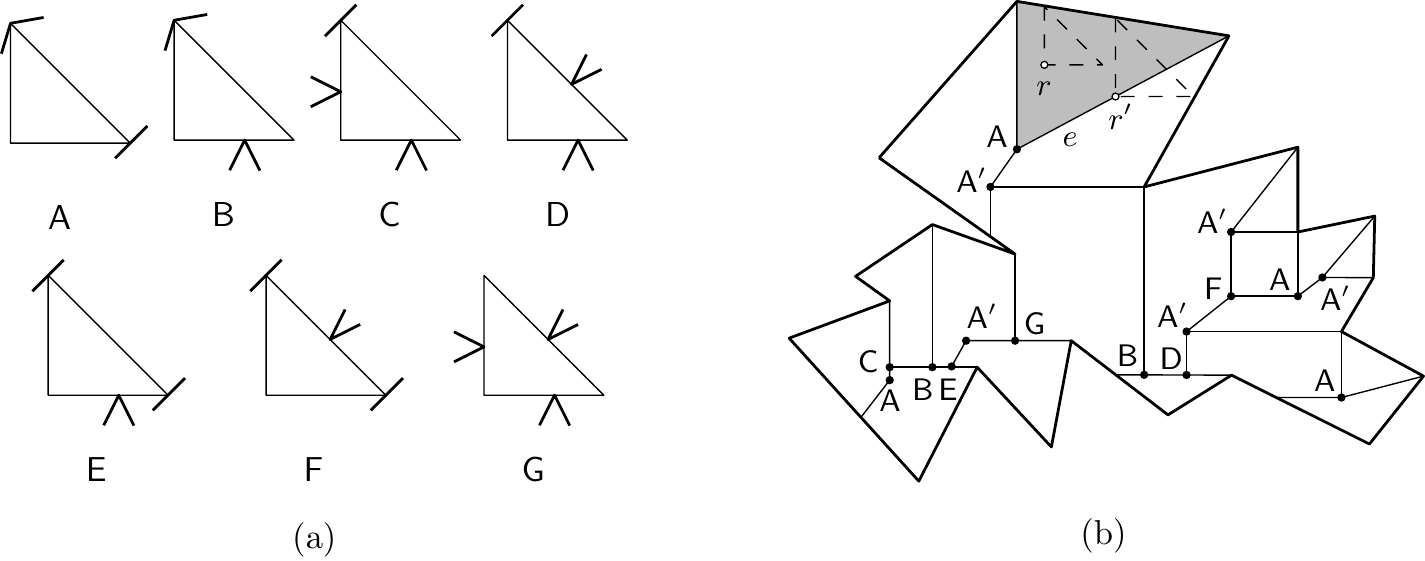}
  \caption{(a) Contact sets of type 3. Symmetric cases are omitted.
  (b) Subdivision of a polygon for $\alpha=\frac{\pi}{2}$ and $\beta=\frac{\pi}{4}$.
  Each vertex of the subdivision is labeled by its corresponding contact set in (a). \textsf{A'} is the symmetric case of \textsf{A}. Point $r$ has a contact set of type 1 and it is in a cell of the subdivision. 
  Point $r'$ has a contact set of type 4 and it is in an edge of the subdivision.}
  \label{fig:contact_subdivision}
\end{figure}

For a point $r$ in the interior of $P$, consider the contact set $C$ of $r$,
which may consist of polygon edges that a corner of $T(r)$ touches
and polygon vertices that a side of $T(r)$ touches in its interior. 
We classify the contact set $C$ of $r$ into four types as follows.
\begin{enumerate}
\item $C$ consists of exactly one edge of $P$.
\item $C$ consists of one or two reflex vertices of $P$ that the side of $T(r)$ 
opposite to $r$ touches in its interior.
\item $C$ belongs to 
one of the configurations shown in Figure~\ref{fig:contact_subdivision}(a)
or their symmetric configurations with respect to the anchor of $T(r)$.
A superset of $C$ also belongs to this type.
\item Other than types 1, 2, and 3.
\end{enumerate}

Observe that each interior point $r$ of $P$ has a contact set, which
belongs to one of the four types defined above.
For a contact set $C$, consider the set $R(C)$ of the points
$r$ in the interior of $P$ such that $C$ remains to be the contact set of 
$T(r)$ under translations and scaling.
Then the classification of contact sets above induces a subdivision of $P$
into cells, edges, and vertices. 
A vertex of the subdivision has degree 1 (endpoint
of a subdivision edge on the boundary of $P$) or larger. 
\iffull
\begin{lemma}
Let $\sd$ be the subdivision of $P$ by $(\alpha,\beta)$, and let $C$ be the contact set of 
a point in the interior of $P$.
Then $R(C)$ is a cell of $\sd$ if $C$ is of type 1 or 2, a vertex of $\sd$ if $C$ is of type 3, 
and an edge of $\sd$ if $C$ is of type 4.
\end{lemma}
\begin{proof}
  Observe that $R(C)$ is a cell of $\sd$ if $C$ is of type 1 or 2, a vertex of 
  $\sd$ if $C$ is of type 3. A contact set $C$ of type 4 is
  (1) a proper subset of a configuration in~\ref{fig:contact_subdivision}(a),
  (2) the set consists of a contact set $C'$ belongs in (1) and
  additional side contacts on the side which contains a side contact in $C'$,
  and (3) the set consists of corner contacts on one corner $c$, except the 
  anchor, and side contacts on both sides incident to $c$.
  Since type 1 and 2 contains all the contact set which contains exactly one 
  element, $C$ contains at least two elements. 
  Then $R(C)$ for $C$ belongs in (1) is line segment. 
  The additional side contacts on the side which contains a side contact does 
  not restrict $R(C)$. Thus, $R(C)$ for $C$ belongs in (2) is also line segment.
  It is obvious that $R(C)$ for $C$ belongs in (3) is line segment.
  Therefore, $R(C)$ is an edge of $\sd$ if $C$ is of type 4.
\end{proof}
\fi
See Figure~\ref{fig:contact_subdivision}(b) that illustrates the subdivision of $P$
for $\alpha=\frac{\pi}{2}$ and $\beta=\frac{\pi}{4}$ into cells, edges and vertices.
Any point $r$ in a cell has the same contact set of type 1 or 2. (The gray cell
has a contact set of type 1.) Any point on an edge of the subdivision has the same contact set of 
type 4. (The edge labeled with $e$ corresponds to a contact set of type 4.)
Each vertex of the subdivision has a contact set of type 3 and is labeled by
its corresponding contact set in Figure~\ref{fig:contact_subdivision}(a).

Observe that an axis-aligned $(\alpha,\beta)$-triangle is not maximal 
if its anchor lies in a cell or edge of $\sd$.
Thus, we have the following lemma.
\begin{lemma}
  \label{lem:max_subdivision}
  Every maximal axis-aligned $(\alpha,\beta)$-triangle has its anchor
  at a vertex of the subdivision of $P$.
\end{lemma}

Now we explain how to construct the subdivision $\sd$ for $P$.
We use a plane sweep algorithm with a sweep line $L$ which has inclination $\pi-\beta$ and
moves downwards. The status of $L$ is the set of rays and edges of $P$
intersecting it, which is maintained in a balanced binary search tree $\tree$ along $L$.
While $L$ moves downwards, the status in $\tree$ changes when $L$ meets particular points.
We call each such particular point \emph{an event point} of the algorithm.
To find and handle these event points, we construct a priority queue $\que$ as the event queue 
which stores the vertices of $P$ in the beginning of algorithm as event points. As $L$ moves
downwards from a position above $P$, 
some event points are newly found and inserted to $\que$ and some event points are removed
from $\que$. 

The invariant we maintain is that at any time during the plane sweep,
the subdivision above the sweep line $L$ has been computed correctly.
Consider the moment at which $L$ reaches a vertex $v$ of $P$.
If $v$ is convex vertex, we add a ray $\ray[\gamma](v)$ to $\sd$ and update $\tree$
and $\que$ if it is contained in $P$ locally around $v$. 
Since every point in $\ray[\gamma](v)$ near $v$ has the same contact set of type 4 consisting
of the two edges, $\gamma$ is determined uniquely by the two edges.
If $v$ is a reflex vertex, we add at most two rays, $\ray[\pi+\alpha](v)$ and $\ray[\pi](v)$ from $v$, to
$\sd$ and update $\tree$ and $\que$ accordingly if the ray is contained in $P$ locally around $v$ and
every point in the ray near $v$ has the same contact set of type 4.
Consider the moment at which $L$ reaches the intersection of a ray with the boundary of $P$.
Then the ray simply stops there.
Consider now the moment at which $L$ reaches the intersection point $x$ of two rays $\eta_1$ and $\eta_2$.
Then the two rays stop at $x$. We add one ray $\eta$ emanating from $x$ to $\sd$ and update $\tree$ and $\que$
accordingly. 
Observe that the contact set of points on $\eta$ near $x$ 
is of type 4 consisting of contact elements of the points in $\eta_1$ and 
$\eta_2$. Thus, the orientation of $\eta$ is uniquely determined in $O(1)$ time. 
If $\eta_1$ or $\eta_2$ emanates from a reflex vertex of $P$, 
$\eta$ makes counterclockwise angle $\pi+\alpha$ or $\pi$ from the positive $x$-axis. 
Imagine we move a point $p$ from $x$ along $\eta$. 
Then the contact set of $p$ may change at some point $q\in\eta$ to another contact set
consisting of contact elements of the points in $\eta_1$ and $\eta_2$.
We call such a point $q$ \emph{a bend point} of $\eta$.  
Again a bend point of a ray can be found in $O(1)$ time. We add $q$ to $\que$ as an event point.
Finally, consider the moment at which $L$ reaches a bend point $q$ of a ray $\eta$ which
emanates from the intersection of two rays. Then $\eta$ stops at $q$. 
We add a new ray $\eta'$ emanating from $q$ to $\sd$ and update $\tree$ 
and $\que$ accordingly. The orientation of $\eta'$ is uniquely determined by the contact elements of 
the points in $\eta$ in $O(1)$ time. Observe that $\eta'$ makes a counterclockwise angle 
other than $\pi+\alpha$ and $\pi$ from the positive $x$-axis.

At each of these event points, we update $\tree$ and $\que$ as follows. At an event point 
where a ray $\eta$ is added to $\sd$,
we insert $\eta$ to $\tree$, compute the event points at which $\eta$ intersects with its neighboring rays
along $L$ and with the boundary of $P$, and add the event points to $\que$. 
At an event where a ray $\eta$ stops, we remove it from $\tree$ and remove the events induced by
$\eta$ from $\que$. We also compute the event at which the two neighboring rays of $\eta$ along $L$
intersect and add them to $\que$. 

After we have treated the last event, we have computed the subdivision of $P$.

\begin{lemma}
  \label{lem:construct_subdivision}
  We can construct the subdivision $\sd$ of $P$  in
  $O(n \log n)$ time using $O(n)$ space.
\end{lemma}
\iffull
\begin{proof}
 First we show that the number of event points in the plane sweep algorithm is $O(n)$.
 A polygon vertex induces at most two rays and generates at most three events at 
 the vertex and two points where the rays intersect the boundary of $P$ for the first time.
 Thus, there are $O(n)$ rays induced by polygon vertices and they generate 
 $O(n)$ event points.
 At an event, either (a) two neighboring rays merge into one at their intersection point or 
 (b) a ray making counterclockwise angle $\pi+\alpha$ or $\pi$ with the positive $x$-axis
 has at most one bend point, and every ray emanating from a bend point makes a counterclockwise angle
 other than $\pi+\alpha, \pi$ from the positive $x$-axis. An event of case (a) generates
 $O(1)$ new event points and the number of rays decreases by 1. Thus, the total number of event points
 of case (a) is $O(n)$. An event point of case (b) generates
 $O(1)$ new events, but only once for a ray making counterclockwise angle $\pi+\alpha$ or $\pi$ 
 with the positive $x$-axis. Again, the total number of event points of case (b) is $O(n)$.
  
  For each event in the plane sweep algorithm, we stop at most two rays and add at most two rays
  to $\sd$ in $O(1)$ time. Then we update $\tree$ and $\que$ accordingly in $O(\log n)$ time since 
  there are $O(n)$ elements in $\tree$ and $O(n)$ events in $\que$.  
  Thus, we can handle an event in $O(\log n)$ time, and we can construct $\sd$ in $O(n \log n)$ time.
  The data structures $\sd$, $\tree$ and $\que$ all use $O(n)$ space.
\end{proof}
\fi
\subsubsection{Computing a largest axis-aligned \texorpdfstring{$(\alpha,\beta)$}{(a,b)}-triangle}
By Lemma \ref{lem:max_subdivision}, it suffices to check all vertices
of $\sd$ to find all maximal axis-aligned $(\alpha,\beta)$-triangles. 
For each vertex $w$ of \sd, the triangle $T(w)$ is a maximal axis-aligned $(\alpha,\beta)$-triangle
satisfying $\seg{w}$ by definition. 
We can compute the area of $T(w)$ in $O(1)$ time by storing $\seg{w}$ 
to $w$ when it is added into $\sd$. Then, we can find a largest axis-aligned $(\alpha,\beta)$ 
triangle by choosing 
a largest one among all maximal axis-aligned $(\alpha,\beta)$-triangles. 
By Lemmas~\ref{lem:max_subdivision} and \ref{lem:construct_subdivision}, we have following theorem.

\begin{theorem}
\label{thm:axis_alpha_beta}
Given a simple polygon $P$ with $n$ vertices in the plane and two angles $\alpha, \beta$, 
we can find a maximum-area $(\alpha,\beta)$-triangle that can be inscribed in $P$
in $O(n \log n)$ time using $O(n)$ space.
\end{theorem}

\subsection{Largest \texorpdfstring{$(\alpha,\beta)$}{(a,b)}-triangles of arbitrary orientations}
\label{sec:alpha_beta_arbitrary}
We describe how to find a largest $(\alpha,\beta)$-triangle of arbitrary orientations that can be
inscribed in a simple polygon $P$ with $n$ vertices. 
We use $\ct$ to denote the coordinate axes obtained by rotating the 
standard $xy$-Cartesian coordinate system by $\theta$ degree counterclockwise around the origin.
We say a triangle $T$ with base $b$ is $\theta$-aligned if $b$ is parallel to the $x$-axis in $\ct[\theta]$.

We use \sdt{\theta} to denote the subdivision of $P$ in $\ct$.
We construct the subdivision $\sdt{\theta}$ of $P$ at $\ct[0]$ using the algorithm in Section~\ref{sec:axis_aligned_fixed_angle}, 
and maintain it while rotating the standard $xy$-Cartesian coordinate axes from
angle 0 to $2\pi$. 
During the rotation, we maintain the combinatorial structure of $\sdt{\theta}$ 
(not the embedded structure $\sdt{\theta}$)
and update the combinatorial structure for each change so that the changes of $\sdt{\theta}$ 
are handled efficiently.
We abuse the notation $\sdt{\theta}$ to refer the combinatorial structure of $\sdt{\theta}$ 
if understood in the context.
For each vertex of $\sdt{\theta}$, we store the function which returns the actual coordinate of the vertex in
the embedded structure $\sdt{\theta}$. Thus, an edge of $\sdt{\theta}$ is determined by the functions
stored at its two endpoints. For each edge of $\sdt{\theta}$, we store the contact set of the points
in the edge.

We say a contact set $C$ is \emph{feasible} at an angle $\theta_0$ if there exists 
a $\theta_0$-aligned $(\alpha,\beta)$-triangle inscribed in $P$ and satisfying $C' \supseteq C$.
For a contact set $C$, consider all angles at which $C$ is feasible. Then these angles
form connected components in $[0,2\pi)$ 
which are disjoint intervals. We call each such interval a feasible interval of $C$.

For a fixed angle $\theta_0$, consider a $\theta_0$-aligned $(\alpha,\beta)$-triangle 
satisfying a contact set $C$. Let $I=[\theta_1,\theta_2]$ be a feasible interval of $C$ 
containing $\theta_0$.
For $\theta_0\in I$, we say a $\theta_0$-aligned $(\alpha,\beta)$-triangle $T$ 
satisfying a contact set $C_1\supseteq C$ is \emph{maximal} for $I$
 if there is no $\theta' \in I$ such that a $\theta'$-aligned
$(\alpha,\beta)$-triangle satisfying a contact set $C_2 \supseteq C$ has larger area.

\subsubsection{Maintaining subdivision under rotations}
\label{sec:maintainDS}

\begin{figure}[ht]
  \centering
  \includegraphics[width=.8\textwidth]{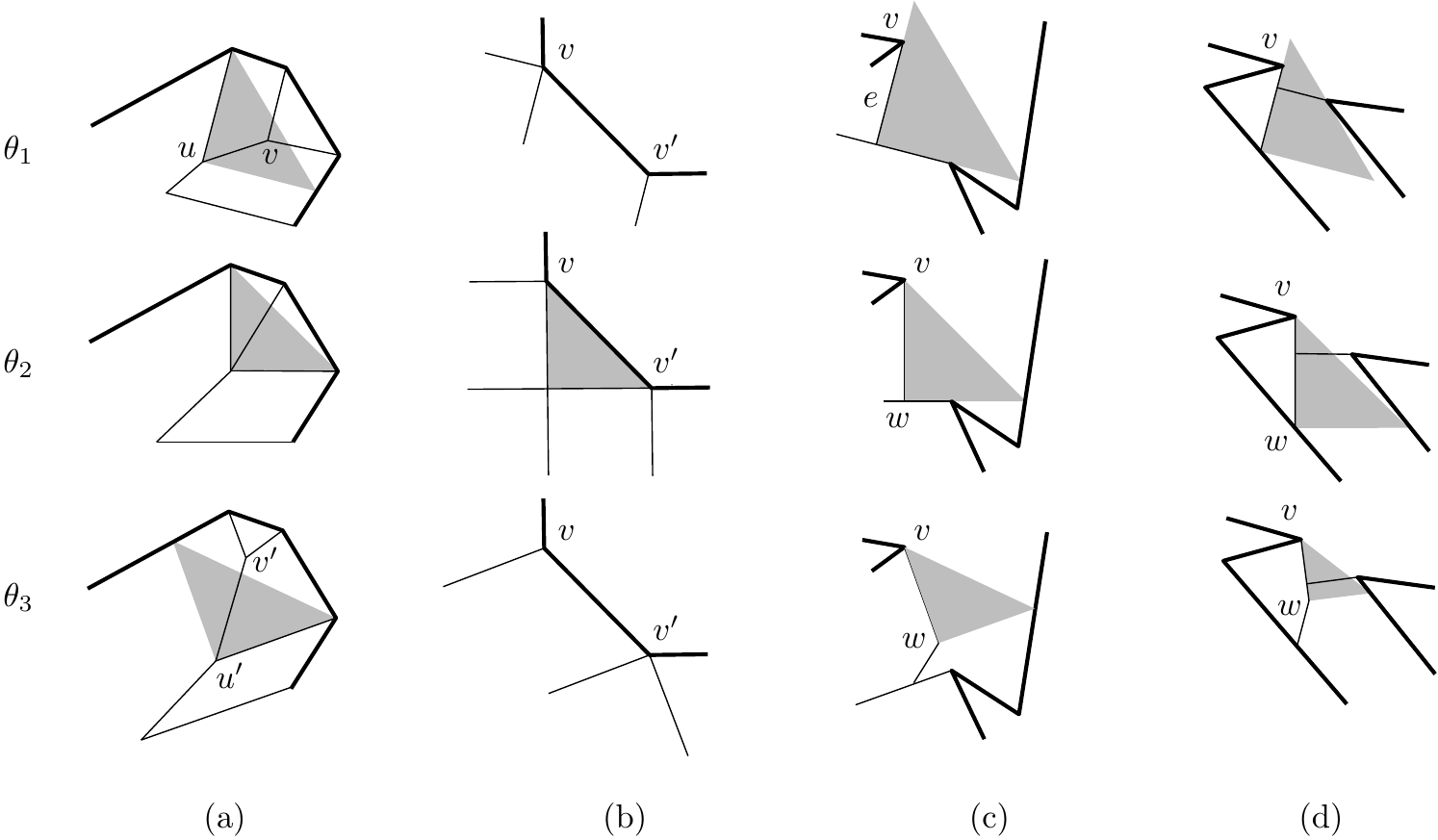}
  \caption{Changes of the subdivision $\sdt{\theta}$ during the rotation with $\theta$ at $\theta_1, \theta_2, \theta_3$ 
  with $\theta_1<\theta_2<\theta_3$.
  (a) An edge event at which edge $uv$ of $\sdt{\theta_1}$ 
  becomes a vertex of $\sdt{\theta_2}$ and then it splits into two with edge 
  $u'v'$ in between them in $\sdt{\theta_3}$.
    (b) A vertex event at which an edge of $\sdt{\theta}$ incident to a reflex vertex $v'$ 
	suddenly appears in $\sdt{\theta_2}$ and an edge of incident to 
	a reflex vertex $v$ suddenly disappears from $\sdt{\theta_3}$. 
	(c) An align event at which a subdivision edge $e$ 
    splits into two edges with vertex $w$ in between.
    (d) A boundary event at which $w$ meets the boundary of $P$ in $\sdt{\theta_2}$.}
  \label{fig:subdivision_events}
\end{figure}
The combinatorial structure of \sdt{\theta} changes during the
rotation. Each change is of one of the following types: 
\begin{itemize}
\item Edge event: an edge of $\sdt{\theta}$ degenerates to
a vertex of $\sdt{\theta}$. Right after the event, the vertex splits into two with an edge
connecting them in $\sdt{\theta}$. See Figure~\ref{fig:subdivision_events}(a). 
\item Vertex event: an edge of $\sdt{\theta}$ incident to a reflex vertex $v$ 
suddenly appears or disappears on $\sdt{\theta}$. 
This event may occur only when an edge $e$ of $P$ incident to $v$
has inclination $0$, $\alpha$, or $\pi-\beta$. See Figure~\ref{fig:subdivision_events}(b). 
\item Align event: an edge $e$ of $\sdt{\theta}$
with inclination $0$ or $\alpha$ splits into two edges with a vertex $w$ of $\sdt{\theta}$ 
in between or two such edges merge into one. 
This event may occur only for vertex $w$ of $\sdt{\theta}$ such that 
the maximal $\theta$-aligned $(\alpha,\beta)$-triangle $T(w)=\tri{wpq}$ (or its symmetric
one) has a reflex vertex 
of $P$ on $q$, a reflex vertex of $P$ on its base $wp$, and an edge on $p$, 
or $T(w)$ has a vertex of $P$ on $p$ and $q$, with one of them being reflex. 
See Figure~\ref{fig:subdivision_events}(c). 
\item Boundary event: a vertex of $\sdt{\theta}$ with degree 2 or larger meets the boundary of 
$P$ or a vertex of $\sdt{\theta}$ with degree 1 meets a vertex of $P$ on the boundary 
of $P$. See Figure~\ref{fig:subdivision_events}(d). 
\end{itemize}

\iffull
\begin{figure}[ht]
  \centering
  \includegraphics[width=.8\textwidth]{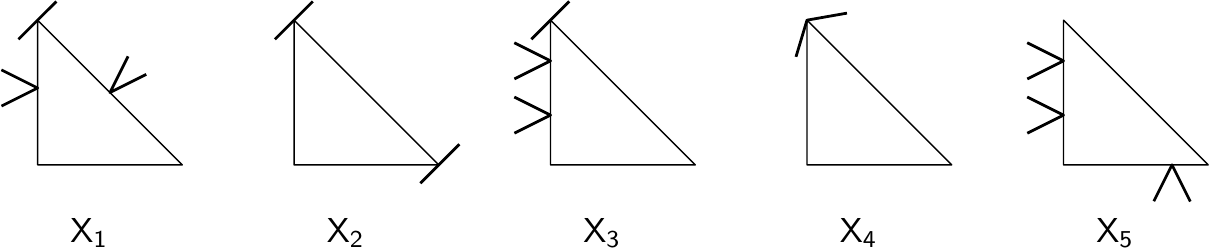}
  \caption{Classification of contact sets of $(\alpha,\beta)$-triangles 
  with anchor at the vertex at which an edge, align, boundary event occurs.}
  \label{fig:contacts type}
\end{figure}
\fi
Recall that the set $\seg{p}$ for a point $p$ on an edge of $\sdt{\theta}$ 
has at least two elements and
$\seg{v}$ for a vertex $v$ of $\sdt{\theta}$ has at least three elements.
We classify all the contact sets of the vertices of $\sdt{\theta}$ 
at which an event (except a vertex event) occurs. 
For a vertex $w$ of $\sdt{\theta}$, $C(w)$ belongs to one of the following types. 
\iffull See Figure~\ref{fig:contacts type} for an illustration. \fi
\begin{itemize}
\item Type $\mathsf{X_1}$: $C(w)$ contains a corner contact at a corner $c$  
and side contacts on both sides incident to $c$.
\item Type $\mathsf{X_2}$: $C(w)$ contains no side contact on a side $s$ and a corner contact on each corner incident to $s$.
\item Type $\mathsf{X_3}$: $C(w)$ contains two side contacts on a side $s$ and a corner contact on a corner incident to $s$.
\item Type $\mathsf{X_4}$: $C(w)$ contains two corner contacts on a corner, i.e. a corner is on a vertex of $P$, except the case that it contains a corner contact on each corner.
\item Type $\mathsf{X_5}$: $C(w)$ contains two side contacts on a side, one side contact on another side,  and no corner contact on the corner shared by the sides.
\end{itemize}
Since there are no two vertices of $\sdt{\theta}$ such that the contact sets of them 
are same, a vertex where an event (except a vertex event) occurs has contact set
containing more than three elements.
Observe that an $(\alpha,\beta)$-triangle $T$ is not maximal
if there is a side $s$ of $T$ such that
no contacts are on both $s$ and corners incident to $s$.
Thus, for a vertex $w$ of $\sdt{\theta}$ at which an event 
(except a vertex event) occurs,
$T(w)$ has no side $s$ such that
no contacts on both $s$ and corners incident to $s$.
Observe that $C(w)$ belongs to a type defined above. 
Thus, the number of events other than vertex events 
is at most the number of maximal $(\alpha,\beta)$-triangles 
which have a contact set of one of the types above.
\iffull
We need the following two technical lemmas to bound the number of $(\alpha,\beta)$-triangles 
satisfying types $\mathsf{X_1}$ and $\mathsf{X_2}$.
\begin{lemma}
  \label{lem:lower_envelope}
  Let $F=\{f_i \mid 1\le i \le n\}$ be a finite family of real value
  functions such that every $f_i$ is of single variable and continuous, 
  any two functions $f_i$ and $f_{i'}$ intersect in their graphs at most
  once, every function $f_i$ has domain $D_i$ of size $d$.  If there is
  a constant $c$ such that $|{\bigcup}{D_i}|=cd$, then the complexity of
  the lower envelope of $F$ is $O(n)$.
\end{lemma}
\begin{proof}
  Let $D_F = {\bigcup}{D_i}$.
  Since there is a constant $c$,
  we can construct a finite partition $A$ of $D_F$ such that each element of
  $A$ has size $d$, except one element of size smaller than or equal to $d$.
  Then every $D_i$
  intersects at most two elements of $A$.
  Let $l_j$ be the left endpoint of $t_j \in A$ and assume $l_j<l_{j'}$ for all $j<j'$.
  Let $L(l_j)$ and $R(l_j)$ be the sets of functions $f'_i$ which are functions $f_i$ restricted to 
  $D_i \cap [l_{j-1},l_j]$ and $D_i \cap [l_j,l_{j+1}]$, respectively.  Since any
  two functions $f'_i, f'_{k} \in L(l_j)$ (or two in $R(l_j)$) 
  intersect in their graphs at most once and their domains have the same start point or end point,
  the sequence of the lower envelope of $L(l_j)$ (or $R(l_j)$) is a Davenport-Schinzel sequence
  of order 2. Then the lower envelope of set $L(l_j)$
  (and of set $R(l_j)$) has complexity $O(k)$, where
  $k=|L(l_j)|$ (or $k=|R(l_j)|$). Since $|\bigcup L(l_j)|=|\bigcup R(l_j)|=O(n)$, 
  the lower envelope of
  $\bigcup L(l_j)$ and $\bigcup R(l_j)$ has complexity $O(n)$.
  
  Now, consider a new partition of $D_F$ obtained by 
  slicing it at every point at which the lower envelopes of $\bigcup L(l_j)$ and
  $\bigcup R(l_j)$ change combinatorially.  
  Since there are $O(n)$ such points and the lower envelope of $F$ restricted to 
  a component of the new partition has constant complexity,
  the complexity of the lower envelope of $F$ is $O(n)$.
\end{proof}
\fi
\iffull
\begin{figure}[ht]
  \centering
  \includegraphics[width=.5\textwidth]{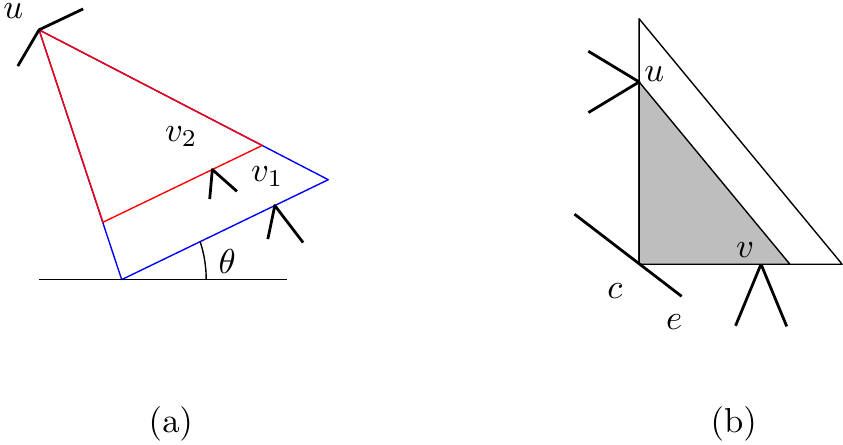}
  \caption{(a) For a fixed angle $\theta$, $T_1\subseteq T_2$ or $T_2\subseteq T_1$
	for any two $\theta$-aligned triangles $T_1$ and  $T_2$ with $u$ at a corner $c$ 
	and a vertex $v_1$ and $v_2$ on the side opposite to $c$, respectively.
  (b) An $(\alpha,\beta)$-triangle $T$ satisfying a contact set $C=\{u,v,e\}$ of type $\mathsf{X_1}$ contains an $(\alpha,\beta)$-triangle $T'$ which 
  shares a corner $c$ with $T$ lying on $e$, has a corner at $u$,
  and has $v$ on the side opposite to $c$.
  }
  \label{fig:triangle contacts}
\end{figure}
\fi
\iffull
\begin{lemma}
  \label{lem:triangle_num}
  For a fixed vertex $u$ of $P$, there are $O(n)$ maximal $(\alpha,\beta)$-triangles $T$
  with a corner at $u$ such that $T$ has a corner on the interior of an edge of $P$
  and has a vertex of $P$ on the side opposite to $u$.
\end{lemma}
\begin{proof}
For a fixed $u$, consider a maximal $(\alpha,\beta)$-triangle $T$ with a corner at $u$ 
such that $T$ has a corner on the interior of an edge $e$ of $P$
  and has a vertex $v$ of $P$ on the side opposite to $u$.
Let $\theta_T$ denote the angle of rotation such that the base of $T$ is parallel 
to the $x$-axis of $\ct[\theta_T]$.

To count such triangles, 
let $\mathcal{T}_v$ denote all nontrivial $(\alpha,\beta)$-triangles with a corner at $u$,
with another vertex of $P$ on the side opposite to $u$. Observe that 
for a fixed angle $\theta$, $T_1\subseteq T_2$ or $T_2\subseteq T_1$
for any two $\theta$-aligned triangles $T_1, T_2\in\mathcal{T}_v$.
Moreover,
no $\theta$-aligned triangle of $\mathcal{T}_v$, except the smallest one, is inscribed
in $P$. See Figure~\ref{fig:triangle contacts}(a).
Let $F_v(\theta)$ be the function that returns the area of the smallest $\theta$-aligned 
triangle of $\mathcal{T}_v$ at angle $\theta$.
Let $\mathcal{T}_e$ denote all nontrivial $(\alpha,\beta)$-triangles with a corner at $u$,
and another corner on the interior of an edge of $P$ which can be inscribed in $P$. 
Observe that for a fixed angle $\theta$, there is at most one $\theta$-aligned triangle
in $\mathcal{T}_e$ as it is required to satisfy the constraint to be inscribed in $P$.
Let $F_e(\theta)$ be the function 
that returns the area of the $\theta$-aligned 
triangle of $\mathcal{T}_e$ at angle $\theta$.
Then $T$ occurs at the angle of an intersection of the graphs of $F_v$ and $F_e$.

Now we count the intersections of the graphs of $F_v$ and $F_e$.
For any vertex $v$ of $P$, the domain of $\area{T_v(\theta)}$ has size $\pi-\alpha-\beta$,
where $T_v(\theta)$ is the $\theta$-aligned $(\alpha,\beta)$-triangle with $u$ at a corner
and $v$ on its side opposite to $u$.
For any pair of vertices $v, v'$ of $P$, $(\area{T_{v_1}(\theta)}$, $\area{T_{v_2}(\theta)}$ have 
at most one intersection. 
The union of all domains of $\area{T_v(\theta)}$'s is $[0, 2 \pi )$.
Thus, by Lemma~\ref{lem:lower_envelope}, the complexity of $F_v$ is $O(n)$.

  For any two edges $e_1,e_2$ of $P$, the portions of the graph of $F_e$
  corresponding to $\area{T_{e_1}(\theta)}$ and $\area{T_{e_2}(\theta)}$ are disjoint 
  since $\ray[\pi+\alpha](u)$ hits only one edge at an angle $\theta$, where $T_e(\theta)$
  denotes the $\theta$-aligned $(\alpha,\beta)$-triangle with $u$ at a corner and $e$ on its anchor. 
  
  For any pair of a vertex $v$ and an edge $e$ of $P$, the graphs of
  $\area{T_v(\theta)}$ and $\area{T_e(\theta)}$ intersect at most twice since 
  the trajectory of point $c$ such that $\angle acv=\beta$ forms a circular arc, 
  and a circular arc intersects a line segment at most twice.
  Thus, the graphs of $F_v$ and $F_e$ intersect $O(n)$ times 
  for a fixed vertex $u$.
  \end{proof}
  \fi
In the following lemma, we bound the number of $(\alpha,\beta)$-triangles 
satisfying one of the types $\mathsf{X_i}$ for $i=1,\ldots,5$.
\begin{lemma}
  \label{lem:subdivision_events}
  There are $O(n^2)$ events that occur to \sdt{\theta} during the rotation.
\end{lemma}
\iffull
\begin{proof}
  Observe that there are $O(n)$ align events.
  Consider an $(\alpha,\beta)$-triangle $T$ satisfying a contact set $C$ belonging to type $\mathsf{X_1}$.
   Then $T$ contains an $(\alpha,\beta)$-triangle $T'$ which 
   shares a corner $c$ with $T$ lying on an edge, has a corner at a vertex $u \in C$,
   and has a vertex of $P$ on the side opposite to $c$.
   See the gray triangle in Figure~\ref{fig:triangle contacts}(b).
   Thus, we find all such triangles $T'$ for every vertex $u$.
   By Lemma~\ref{lem:triangle_num}, there are $O(n)$ such triangles for a vertex $u$ of $P$,
   and in total $O(n^2)$ $(\alpha,\beta)$-triangles satisfying the contact sets of type $\mathsf{X_1}$.
    
   Consider an $(\alpha,\beta)$-triangle $T$ satisfying a contact set $C$ belonging to type $\mathsf{X_1}$.
   Then there is at least one vertex $v$ of $P$ such that
   the boundary of $P$ contains $v$.
   It is obvious that the number of $(\alpha,\beta)$-triangles with corner at $v$ and satisfies a contact set of type $\mathsf{X_2}$.
   Thus, assume $v$ is side contact of $T$.
   Let $e_1$ and $e_2$ be two edges that $C$ contains
   and $x$ be the vertex or edge of $P$ which is contained in $C$ and distinct to $v$, $e_1$, and $e_2$.
   We can consider two $\theta$-aligned $(\alpha,\beta)$-triangles $T_1(\theta)$ and $T_2(\theta)$ which satisfies
   the contact set consists of $\{v,x,e_1\}$ and $\{v,x,e_2\}$ respectively.
   Similar to Lemma~\ref{lem:triangle_num}, we consider the number of intersection of the graphs of $\area{T_1(\theta)}$ and $\area{T_2(\theta)}$. 
 Then we can prove that the number of such intersection is $O(n^2)$.
 Thus, the number of $(\alpha,\beta)$-triangles satisfying a contact set of $\mathsf{X_4}$ is $O(n^2)$. 
  
 Observe that there are constant number contact set $C$ of type $\mathsf{X_3}$ 
   if two vertices of $P$ which lies on a same side is given.
   And each $C$ has constant number of feasible orientations.
   The number of contact sets of type $\mathsf{X_4}$, 
   except the case of contact set contain two side contacts $v_1$ and $v_2$ on the side opposite to 
   the corner which is on a vertex $u$ of $P$, 
   is also $O(n^2)$ because of the same reason.
   The number of contact set of excepted case can be considered in a similar way to Lemma~\ref{lem:triangle_num},
   by considering two triangles $T_{v_1}$ and $T_{v_2}$
   such that both $T_{v_1}$ and $T_{v_2}$ has a corner at $u$ and 
    $v_i$ is on the side opposite to $u$ of $T_{v_i}$. 
   Thus, the number of $(\alpha,\beta)$-triangles satisfying a contact set of $\mathsf{X_4}$ is $O(n^2)$. 
   
   An $(\alpha,\beta)$-triangle $T$ satisfying a contact set of type $\mathsf{X_5}$ contains an
   $(\alpha,\beta$-triangle $T'$ which has a corner at $v$, where $v$ is a side contact of $T$ 
   similar to type $\mathsf{X_1}$. Since the contact set of $T'$ is of type $\mathsf{X_4}$,
    the number of contact set of $\mathsf{X_5}$ is $O(n^2)$. 
   
    Therefore, the number of events that occurs to $\sdt{\theta}$ during the rotation is also $O(n^2)$.
 \end{proof}
\fi
To capture the combinatorial changes and maintain \sdt{\theta} during
the rotation, we construct and maintain the following data structures: (1) An event
queue $\que$ which is a priority queue
that stores events indexed by their angles. 
(2) A planar graph representing the combinatorial structure of
\sdt{\theta}.
(3) For each edge $e$ of $P$, a balanced binary search trees $\tree(e)$.
The tree store degree-1 vertices of \sdt{\theta} in order along $e$.

In the initialization, we construct \sdt{0}, and then $\tree(e)$ for each edge of $P$.  
Then we initialize $\que$ with the events defined by the vertices of $\sdt{0}$ and
the vertex events and boundary events defined by the polygon vertices. 
For each vertex $v$ of \sdt{0}, we compute the angle at which
$v$ and a neighboring vertex of $v$ meet (edge event), an edge incident to 
$v$ splits into two (align event), or 
$v$ meets a polygon vertex (boundary event). 
These angles can be computed in $O(1)$ time for each
$v$ using the contacts corresponding to $v$.
Then the size of $\que$ is $O(n)$ which can be constructed in $O(n\log n)$ time.

%

We update each data
structure whenever an event occurs.  Note that each event changes
a constant number of elements of \sdt{\theta}, creates 
a constant number of events to $\que$, 
and removes a constant number of events from $\que$. 
Thus we can update the subdivision in $O(1)$ time, 
and the tree $\tree(e)$ in $O(\log n)$ time for edge $e$ where a boundary event occurs.

\begin{lemma} \label{lem:subdivision.construction}
  We can construct the subdivision of $P$ and maintain it 
  in $O(n^2 \log n)$ time using $O(n)$ space during rotation.
\end{lemma}
\iffull
\begin{proof}
  By Lemma~\ref{lem:construct_subdivision}, we can construct \sdt{0}
  in $O(n\log n)$ time using $O(n)$ space.  We can construct all
  the data structures for maintaining \sdt{\theta} in $O(n\log n)$ time
  using $O(n)$ space.  By Lemma~\ref{lem:subdivision_events}, there
  are $O(n^2)$ events during the rotation. Each event, except vertex events,
  can be handled in $O(\log n)$ time.
  For each vertex event, we reconstruct $\sdt{\theta}$, $\que$ and $\tree(e)$ for
  each edge $e$ of $P$.
  Since there are $O(n)$ vertex events and it takes $O(n\log n)$ time for the reconstruction,
  it takes $O(n^2 \log n )$ time to handle all vertex events.
   The space complexity remains to be $O(n)$ since the complexity of the data structures 
   is $O(n)$.
\end{proof}
\fi
\subsubsection{Computing a largest \texorpdfstring{$(\alpha,\beta)$}{(a,b)}-triangle}
Whenever a vertex $w$ appears on the subdivision or
$\seg{w}$ changes at $\theta_0$ by an event occurring at angle $\theta_0$, 
we store $\theta_0$ at $w$.
Whenever $\seg{w}$ changes or $w$ disappears
from the subdivision at angle $\theta_1$, we compute the largest
$\theta$-aligned $(\alpha,\beta)$-triangle satisfying
$\seg{w}$ for $\theta \in [\theta_0,\theta_1]$, where
$\theta_0$ is the angle closest from $\theta_1$ at which $w$ appears or $\seg{w}$ changes with 
$\theta_0<\theta_1$.
We do this on every vertex $w$ of the subdivision, and then return the largest triangle
among the triangles on the vertices.
We can compute the largest one among all 
$\theta$-aligned  $(\alpha,\beta)$-triangles satisfying $\seg{w}$ 
for $\theta\in [\theta_0,\theta_1]$
in $O(1)$ time using the area function of the $\theta$-aligned $(\alpha,\beta)$-triangle satisfying $\seg{w}$.
Thus, from Lemma~\ref{lem:subdivision.construction}, we have the following theorem.

\begin{theorem}
  Given a simple polygon $P$ with $n$ vertices in the plane and two angles $\alpha, \beta$,
  we can find a maximum-area $(\alpha,\beta)$-triangle 
  that can be inscribed in $P$ in $O(n^2 \log n)$ time using $O(n)$ space.
\end{theorem}

\section{Largest \texorpdfstring{$\alpha$}{a}-triangles in a Simple Polygon}
In this section, we compute a largest $\alpha$-triangle that can be inscribed
in a simple polygon $P$ with $n$ vertices in the plane.
Without loss of generality, 
we assume no three vertices of $P$ are collinear.
\iffull \subsection{Largest axis-aligned \texorpdfstring{$\alpha$}{a}-triangles} \fi
We consider a largest axis-aligned $\alpha$-triangle that can be inscribed in $P$. 
\iffull We use $\sdt{\theta}$ to denote the subdivision of $P$ defined for two angles $\alpha$ and $\theta$.
We say a contact set $C$ is \emph{feasible} at an angle $\theta_0$ if 
there exists an axis-aligned $(\alpha,\theta_0)$-triangle inscribed in $P$
satisfying $C'\supseteq{C}$.  
For a feasible interval $I$ of a contact set $C$ and $\theta_0\in I$, 
we say an axis-aligned $(\alpha,\theta_0)$-triangle $T$ satisfying a contact set $C_1 \supseteq C$ is \emph{maximal} in $I$ if 
there is no $\theta' \in I$ such that an axis-aligned $(\alpha,\theta')$-triangle satisfying a contact set $C_2 \supseteq C$ has larger area.     
The point at which $\ray[\theta](p)$ meets the boundary of 
a simple polygon $Q\subseteq P$ for the first time other than its source point 
is called the \emph{foot of} $\ray[\theta](p)$ 
\emph{on} $Q$ and denoted by $\foot[\theta](Q,p)$.
In an $\alpha$-triangle, we say the side opposite to the anchor is 
\emph{the diagonal} of the triangle. 
For a point $p\in P$, we define the \emph{visibility region} of $p$ as
$\vis(p)=\{x\in P\mid px\subset P\}$.  For a ray $\eta$ and an angle
$\theta$, we define the $\theta$\emph{-visibility region} $\visl[\theta](\eta)$ of $\eta$ 
as the set of points $x$ in the segment $y\foot[\theta](P,y)$ contained in $P$ 
for every  $y\in \eta$.

\fi
First, we find the largest axis-aligned $(\alpha,\frac{\pi-\alpha}{2})$-triangle $T$
using the algorithm in Section~\ref{sec:axis_aligned_fixed_angle}.
Let $d$ be the diameter of $P$.
Since every side of a triangle inscribed in $P$ has length less than or equal to $d$,
any $(\alpha,\theta)$-triangle that can be inscribed in $P$ 
has area less than or equal to $\frac{d^2\sin\alpha\sin\theta}{2\sin(\alpha+\theta)}$.
Thus it suffices to consider $(\alpha,\theta)$-triangles for $\theta>\theta_T$ to find a largest axis-aligned $\alpha$-triangle,
where $\theta_T$ satisfies $\frac{d^2\sin\alpha\sin\theta_T}{2\sin(\alpha+\theta_T)}=\area{T}$.
To find a largest axis-aligned $\alpha$-triangle, 
we choose an angle $\theta_0 \le \theta_T$, construct $\sdt{\theta_0}$ and maintain it while increasing $\theta$ from $\theta_0$ to $\pi-\alpha$.
Note that we can compute $T$ in $O(n\log n)$ time by Theorem~\ref{thm:axis_alpha_beta}
 and we can find $\theta_0$ in $O(1)$ time.
\iffull
\begin{figure}[ht]
  \begin{center}
    \includegraphics[width=0.8\textwidth]{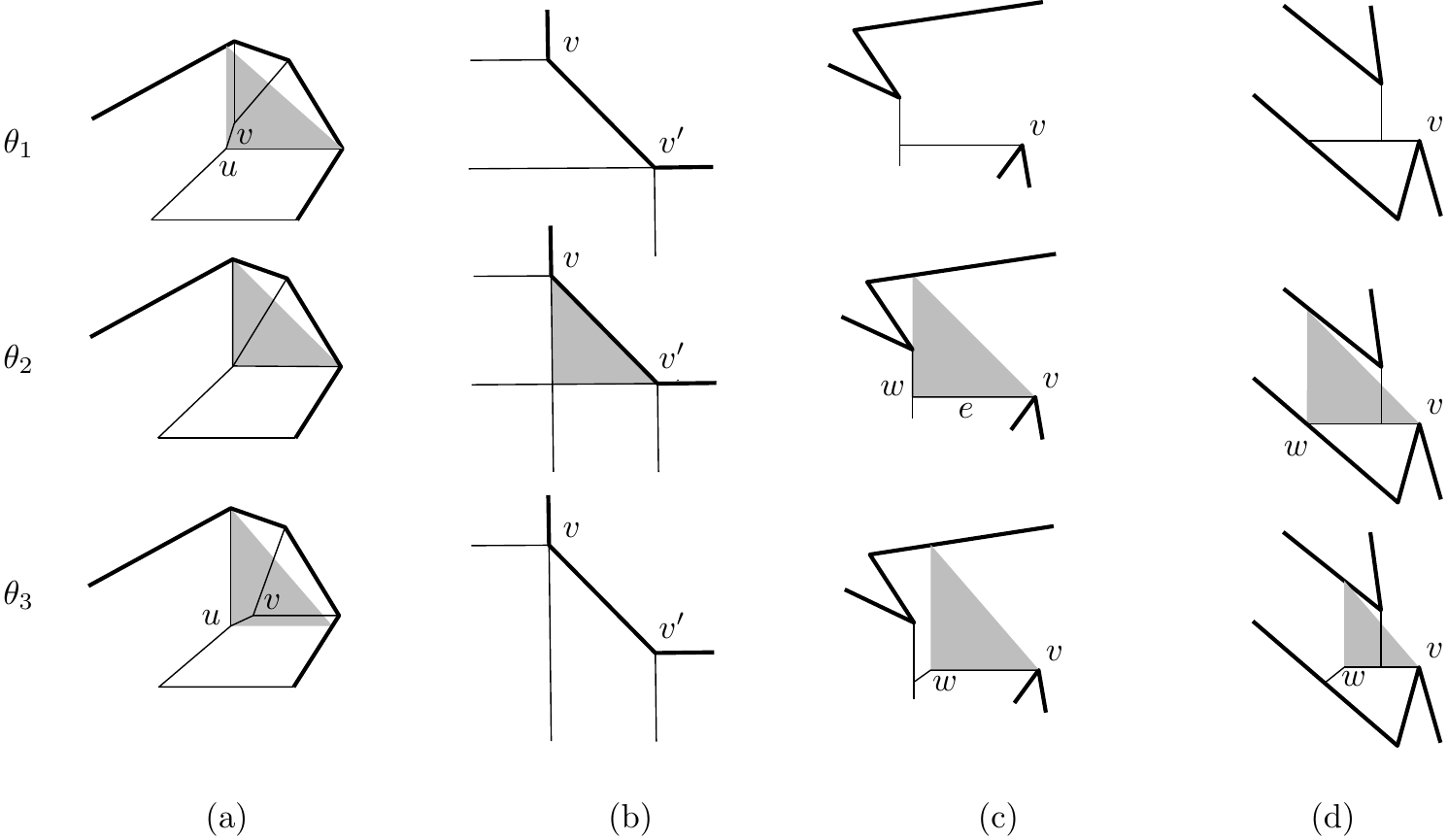}
  \end{center}
  \caption{Changes of the subdivision $\sdt{\theta}$ while increasing $\theta$ at $\theta_1,\theta_2,\theta_3$
  with $\theta_1<\theta_2<\theta_3$. 
  (a) An edge event. 
  (b) A vertex event. 
  (c) An align event. 
  (d) A boundary event.
  }
  \label{fig:subdivision_events_alpha}
\end{figure}
\fi
The combinatorial structure of $\sdt{\theta}$ changes while increasing $\theta$.
We use the definitions for the combinatorial changes of $\sdt{\theta}$
in Section~\ref{sec:alpha_beta_arbitrary}:
 edge, vertex, align and boundary events. 
\iffull See Figure~\ref{fig:subdivision_events_alpha}. \fi
Note that a vertex event occurs at 
$\theta=\pi-\gamma$, where $\gamma$ is the inclination of an edge of $P$. 
\iffull
\subsubsection{The number of edge and vertex events}
\fi
We count all $\alpha$-triangles satisfying a contact set 
of a vertex at which an event occurs.
\iffull
\begin{figure}[ht]
  \begin{center}
    \includegraphics[width=\textwidth]{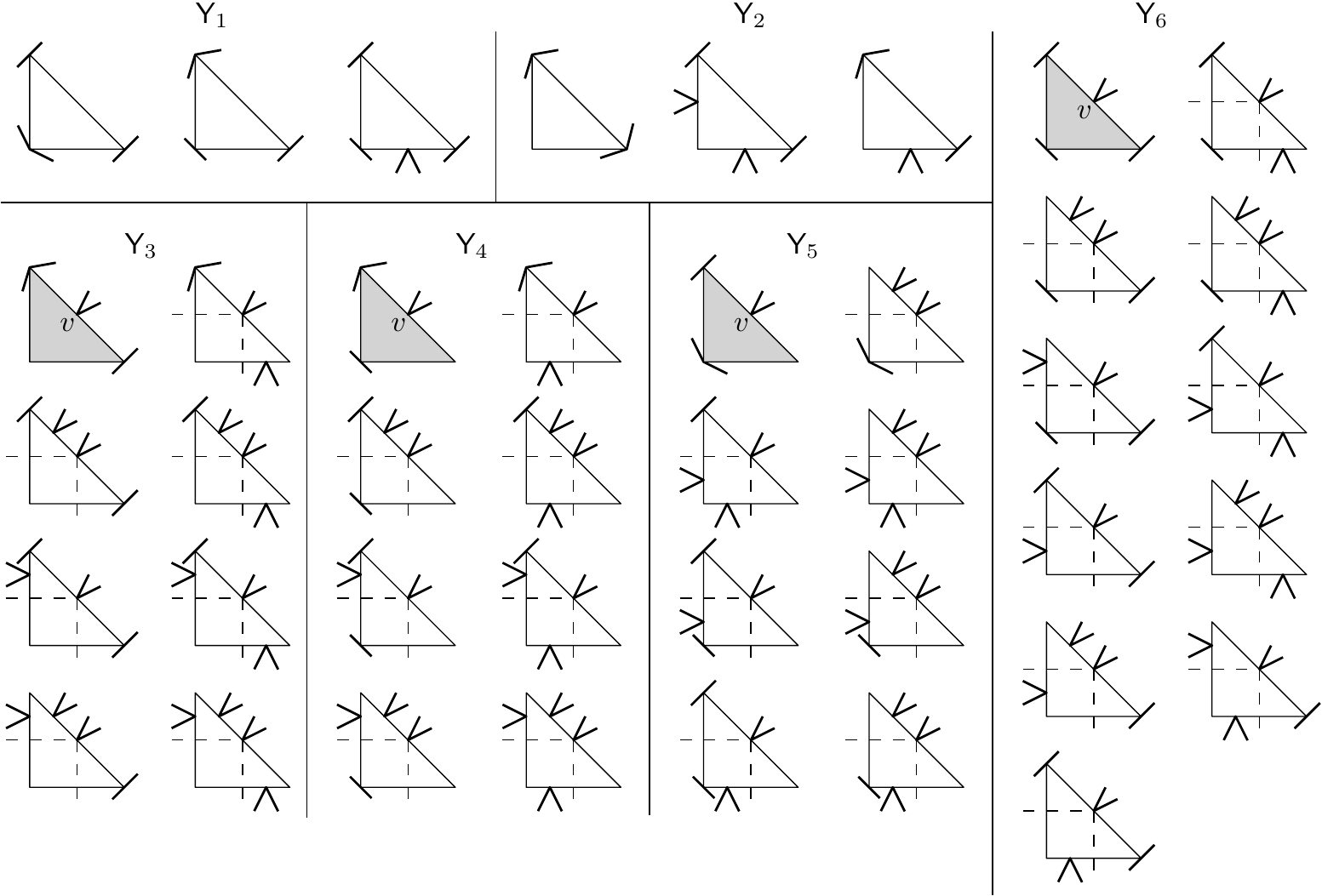}
  \end{center}
  \caption{Classification of contact sets. Symmetric cases are omitted.}
  \label{fig:4-contact_classification}
\end{figure}
\fi
\iffull
Figure~\ref{fig:4-contact_classification} 
shows a classification of all contact sets of vertices of $\sdt{\theta}$ at which 
an event (except vertex events) occurs. 
No contact set of types $\mathsf{Y_1}$ or $\mathsf{Y_2}$ 
contains a diagonal contact, while
contact sets of other types contain a diagonal contact.
\fi
\iffull
We first show that the number of $\alpha$-triangles satisfying contacts sets of
types $\mathsf{Y_1}$ and $\mathsf{Y_2}$ is $O(n^2)$.
Then we show that the number of $\alpha$-triangles satisfying contact sets of 
other types is also $O(n^2)$.
\begin{lemma}
  \label{lem:typeY12}
  There are $O(n^2)$ axis-aligned $\alpha$-triangles satisfying a contact set of 
  type $\mathsf{Y_1}$ or $\mathsf{Y_2}$.
\end{lemma}
\begin{proof}
  Any contact set of type $\mathsf{Y_1}$ or $\mathsf{Y_2}$ contains vertices of $P$.
  For a vertex $v$ of $P$, there is at most one axis-aligned $\alpha$-triangle 
  which satisfies a contact set of type $\mathsf{Y_1}$ containing $v$.
  For any two polygon vertices $v_1$  and $v_2$, there is at most one axis-aligned $\alpha$-triangle 
  which satisfies a contact set of type $\mathsf{Y_2}$ containing $v_1$ and $v_2$.
  Thus, there are $O(n^2)$ axis-aligned $\alpha$-triangles satisfying
  a contact set of type $\mathsf{Y_1}$ or $\mathsf{Y_2}$.
\end{proof}

To count all axis-aligned $\alpha$-triangles satisfying a contact set of a type 
$\mathsf{Y}_i$ for $i=3,\ldots,6$,
we consider such triangles whose diagonal contains a reflex vertex of $P$.
For a reflex vertex $v$ of $P$, let
$P'(v)=\big(\vis(v)\cap \visl[\alpha](\ray[\pi](v))\big)\cup
\big(\vis(v)\cap \visl[0](\ray[\pi+\alpha](v))\big)\cup
\big(\visl[\pi+\alpha](\ray[\pi](v))\cap \visl[\pi](\ray[\pi+\alpha](v))\big)$. 
See Figure~\ref{fig:preprocessing}(a).
Then, every axis-aligned $\alpha$-triangle with $v$ on its diagonal is inscribed in $P'(v)$.
The gray triangles and their contacts in Figure~\ref{fig:4-contact_classification} 
show axis-aligned $\alpha$-triangles $T$ satisfying a contact set of a type 
$\mathsf{Y}_i$ for $i=3,\ldots, 6$ 
with their diagonals containing $v$ and their contact sets with respect to $P'(v)$.
\fi
\iffull
\begin{figure}[ht]
  \begin{center}
    \includegraphics[width=\textwidth]{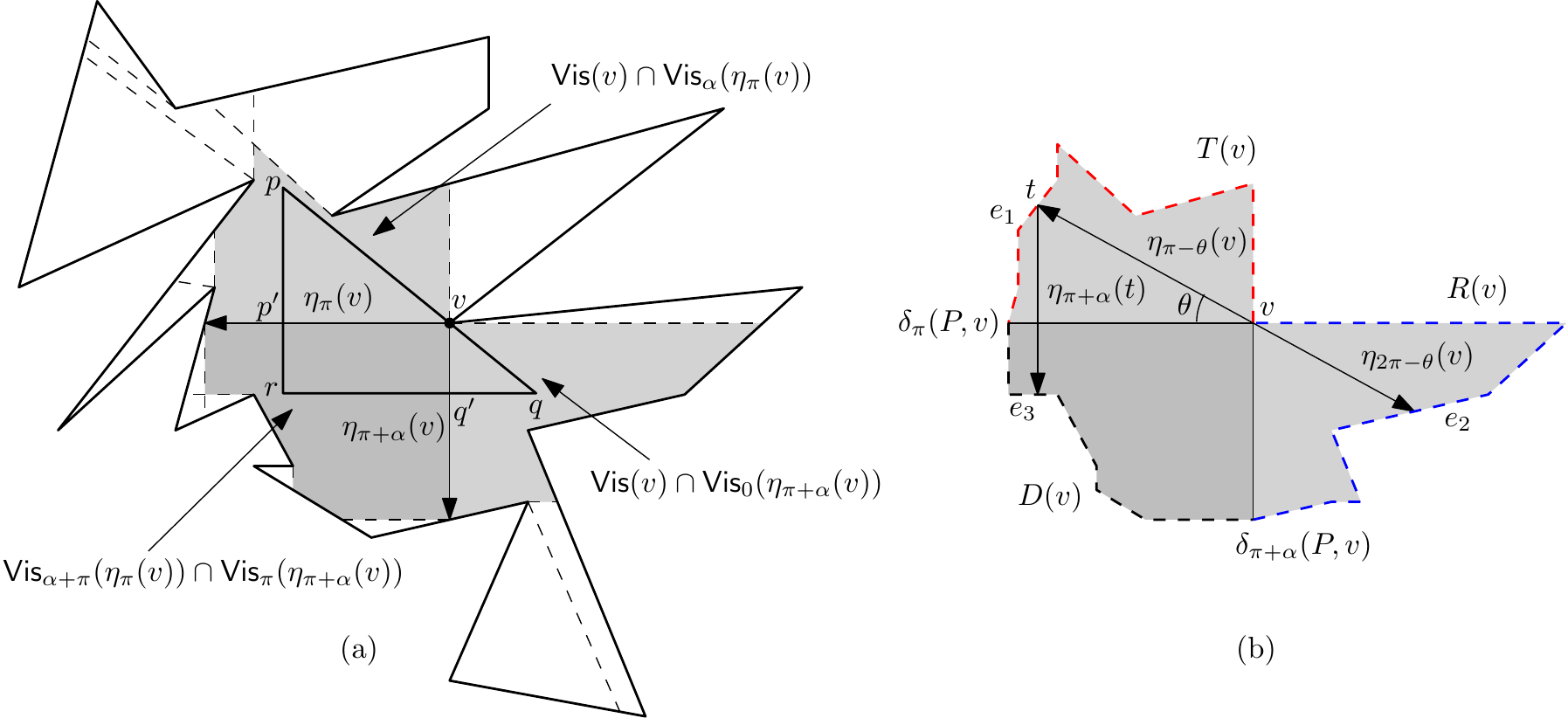}
  \end{center}
  \caption{(a) Preprocessing for computing $\alpha$-triangles satisfying a contact set of type $\mathsf{Y_i}$ for $i=3,\ldots, 6$.
    (b) A quadruplet $(v,e_1,e_2,e_3)$ that can be a contact set of type $\mathsf{Y_6}$. Let $t=\delta_{\pi-\theta}(P'(v),v)$.}
  \label{fig:preprocessing}
\end{figure}
\fi
\iffull
\begin{lemma}
  \label{lem:not_typeY12}
  For a reflex vertex $v$ of $P$, the number of axis-aligned $\alpha$-triangles satisfying
  a contact set of type $\mathsf{Y}_i$ for $i=3,\ldots,6$
  and containing $v$ on their diagonals is $O(n)$.
\end{lemma}
\begin{proof}
  Let $\tc, \rc$, and $\dc$ be the polygonal chains of $P'(v)$ from
  $\foot[\pi](P,v)$ to $v$, from $v$ to $\foot[\pi+\alpha](P,v)$,
  and from $\foot[\pi+\alpha](P,v)$ to $\foot[\pi](P,v)$ in
  clockwise, respectively. See Figure~\ref{fig:preprocessing}(b).

  For a vertex $r$ of $\tc$, there is at most one axis-aligned $\alpha$-triangle 
  which satisfies a contact set of type $\mathsf{Y_3}$,
  contains $v$ on its diagonal, and contains a corner at $r$.
  Thus, there are $O(n)$ axis-aligned 
  $\alpha$-triangles satisfying a contact set of type $\mathsf{Y_3}$.
  
  By the same reason, there are a constant number of $\alpha$-triangles 
  satisfying a contact set of type $\mathsf{Y_4}$ for each vertex of $\tc$, 
  and a constant number of $\alpha$-triangles 
  satisfying a contact set of type $\mathsf{Y_5}$ for each vertex of $\dc$.

 Now we count the number of quadruplets of elements including $v$ 
 that can form a contact set of type $\mathsf{Y_6}$.
  This quadruplet changes only if $\ray[\pi-\theta](v)$, 
  $\ray[2\pi-\theta](v)$, $\ray[\pi+\alpha](\foot[\pi-\theta](P'(v),v))$ meets 
  another vertex of $P'(v)$.
  Since $\ray[\pi-\theta](v)$ and $\ray[2\pi-\theta](v)$ rotate clockwise around $v$, and
  $\ray[\pi+\alpha](\foot[\pi-\theta](P'(v),v))$ moves rightwards while increasing $\theta$,
  each of the rays meets a vertex of $P'(v)$ at most once.
  See Figure \ref{fig:preprocessing}(b).
  Thus, there are $O(n)$ contact sets of type $\mathsf{Y_6}$ containing $v$ as a diagonal contact. 
  For a contact set $C=\{v,e_1,e_2,e_3\}$ of type $\mathsf{Y_6}$ containing $v$ as a diagonal contact,
  we can find the number of axis-aligned $\alpha$-triangles satisfying $C$ in 
  a way similar to the proof on type $\mathsf{X_2}$ of Lemma~\ref{lem:subdivision_events}. 
  If there is an axis-aligned $(\alpha,\theta_T)$-triangle $T$ satisfying $C$, then there are 
  two axis-aligned $(\alpha,\theta)$-triangles $T'(\theta)$ and $T''(\theta)$ 
  such that $T'(\theta)$ satisfies $\{v,e_1,e_2\}$ and $T''(\theta)$ satisfies 
  $\{v,e_2,e_3\}$ and $\area{T}=\area{T'(\theta_T)}=\area{T''(\theta_T)}$.
  Since each area function consists of a constant number of trigonometric 
  functions with period $2\pi$, there are $O(1)$ distinct angles at which 
  $\{v,e_1,e_2,e_3\}$ is feasible. 
  
  Therefore, the number of axis-aligned $\alpha$-triangles satisfying 
  a contact set of each type in Figure~\ref{fig:4-contact_classification}, 
  except types $\mathsf{Y_1}$ and $\mathsf{Y_2}$,
  and containing $v$ on their diagonals is $O(n)$ for each $v$.
\end{proof}
\fi
\iffull
By Lemmas~\ref{lem:typeY12} and \ref{lem:not_typeY12}, we can conclude
with the following lemma.
\fi
\begin{lemma}
  \label{lem:alpha_events}
  There are $O(n^2)$ events, except vertex events, that occur to $\sdt{\theta}$ 
  while increasing $\theta$ from $\theta_0$ to $\pi-\alpha$.
\end{lemma}

\iffull \subsubsection{Maintaining the subdivision while increasing \texorpdfstring{$\theta$}{t}} \fi
To capture the combinatorial changes and maintain $\sdt{\theta}$ while increasing $\theta$, 
we maintain the same data structures defined in 
Section~\ref{sec:maintainDS}, but with different 
equations for computing angles at which an event occurs.
By following the same initialization and update steps in Section~\ref{sec:maintainDS},
we can construct and maintain subdivision $\sdt{\theta}$ while increasing $\theta$.
By Lemmas~\ref{lem:construct_subdivision} and~\ref{lem:alpha_events}, we have following lemma.

\begin{lemma}
  We can construct the subdivision $\sdt{\theta}$ of $P$ and maintain it 
  in $O(n^2 \log n)$ time using $O(n)$ space 
  while increasing $\theta$ from $\theta_0$ to $\pi-\alpha$.
  \label{lem:subdivision_maintain_alpha}
\end{lemma}

\iffull \subsubsection{Computing the largest axis-aligned \texorpdfstring{$\alpha$}{a}-triangles} \fi
If an event occurs at a vertex $w$ of the subdivision and an angle $\theta$,
$C(w)$ changes. 
We find all feasible intervals of contact sets while maintaining the subdivision.
We can compute the maximal axis-aligned $\alpha$-triangle satisfying a contact set 
in $O(1)$ time if a feasible interval is given.
Thus, we have the following theorem by Lemma~\ref{lem:subdivision_maintain_alpha}.

\begin{theorem}
  Given a simple polygon $P$ with $n$ vertices in the plane and an angle $\alpha$,
  we can compute a maximum-area axis-aligned $\alpha$-triangle that can be inscribed in $P$ 
  in $O(n^2 \log n)$ time using $O(n)$ space.
\end{theorem}

\iffull \subsection{Largest \texorpdfstring{$\alpha$}{a}-triangles of arbitrary orientations} \fi
To find a largest $\alpha$-triangle of arbitrary orientations, we follow the 
approach by Melissaratos et al.~\cite{melissaratos1992shortest} in computing a 
largest triangle with no restrictions in a simple polygon, with some 
modification. 
\iffull Their algorithm considers all triangles but we consider $\alpha$-triangles only.
They divide the cases by the number of corners of the triangle lying on the boundary of $P$.
They denote by $m$-case the case that $m$ corners of a triangle lie on the boundary of $P$,
for $m=0,1,2,3$.   

Consider a contact set $C$ consisting of at most three elements.
Then any $\alpha$-triangle 
satisfying $C$ can be enlarged into another $\alpha$-triangle 
while satisfying $C$. So, the contact set of a largest $\alpha$-triangle 
consists of at least four elements. 
Also, for triangles of the 0-case, their contact sets consists of five elements 
and the side opposite to the fixed angle corner contains two side contacts.
See Figure~\ref{fig:a_cases}(i). 
This can be proved in a way similar to Lemma 
6.3 in \cite{melissaratos1992shortest}. 

Figure~\ref{fig:a_cases} illustrates the classification
of contact sets of the largest $\alpha$-triangles for each case.
\begin{figure}[ht]
    \begin{center}
      \includegraphics[width=\textwidth]{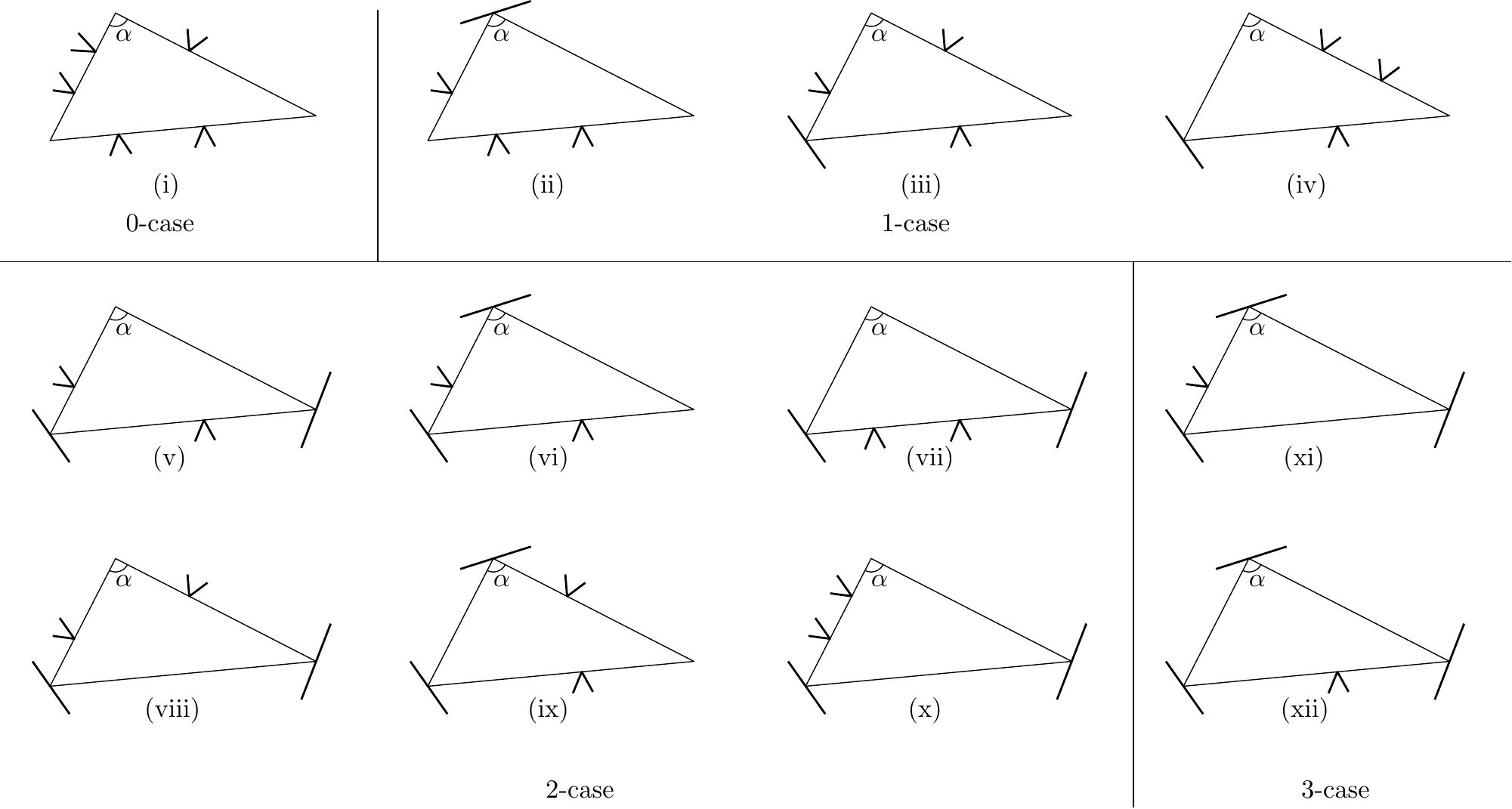}
    \end{center}
    \caption{Classification of the largest $\alpha$-triangles}
    \label{fig:a_cases}
  \end{figure}

\begin{figure}[ht]
  \begin{center}
    \includegraphics[width=0.7\textwidth]{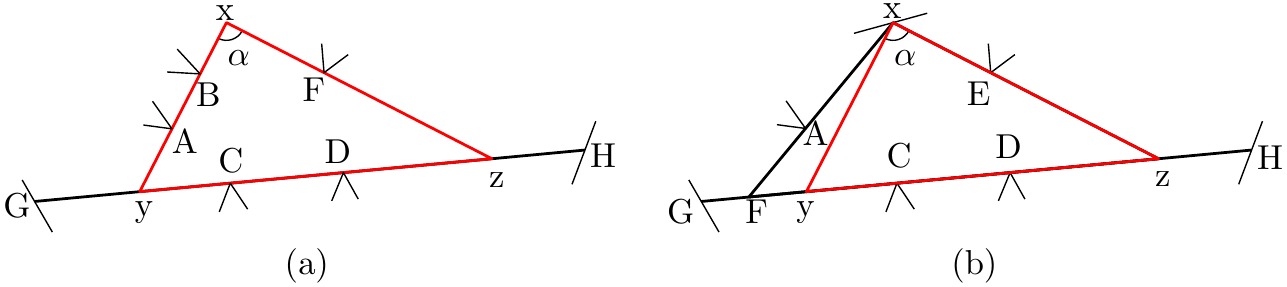}
  \end{center}
  \caption{(a) Configuration (i) in Figure~\ref{fig:a_cases}. (b) Configuration (ii) in Figure~\ref{fig:a_cases}.}
  \label{fig:0_case}
\end{figure}

For cases (i) and (ii) of Figure~\ref{fig:a_cases}, we fix two reflex vertices $C$ and $D$ 
lying on the same side and find the points $G,H$ where
the line containing $CD$ intersects the boundary of $P$
with $GH\subset P$. See Figure~\ref{fig:0_case}(a) and 
(b). Then by walking on the shortest-path maps of $G$ and $H$ along the 
boundary of $P$ in a way similar to \cite{melissaratos1992shortest}, 
we can compute a largest $\alpha$-triangle for each case using $O(n^4)$ time.

For cases (iv), (vii), (x) of Figure~\ref{fig:a_cases}, we find the largest $\alpha$-triangle 
satisfying a contact set $C$ without any restriction on the boundary of $P$. If the 
largest $\alpha$-triangle satisfying $C$ is not inscribed in $P$, 
then the maximum area is achieved at the boundary angles of the feasible intervals 
of $C$, which are handled for other cases.

A contact set $C$ of the remaining cases contains at most one side contact for each side of 
the $\alpha$-triangle satisfying $C$.
For each of these cases, we find a largest $\alpha$-triangle that can be inscribed in $P$
using the method in~\cite{melissaratos1992shortest}.
We decompose the problem into $O(n^4)$ simple optimization problems 
and add a constraint such that one interior angle of triangles must be $\alpha$
to each optimization problem.
Since the original optimization problem can be solved in constant time, our problem can also 
be solved in $O(1)$ time.
\fi
\begin{theorem}
  Given a simple polygon $P$ with $n$ vertices in the plane and an angle $\alpha$,
  we can compute a maximum-area $\alpha$-triangle inscribed in $P$ 
  in $O(n^4)$ time using $O(n)$ space.
\end{theorem}
\bibliography{reference.bib} \bibliographystyle{plainurl}
\end{document}